\PassOptionsToPackage{unicode}{hyperref}
\PassOptionsToPackage{hyphens}{url}
\PassOptionsToPackage{dvipsnames,svgnames,x11names}{xcolor}
\documentclass[
  12pt,
]{article}
\usepackage{amsmath,amssymb}
\usepackage{iftex}
\ifPDFTeX
  \usepackage[T1]{fontenc}
  \usepackage[utf8]{inputenc}
  \usepackage{textcomp} 
\else 
  \usepackage{unicode-math} 
  \defaultfontfeatures{Scale=MatchLowercase}
  \defaultfontfeatures[\rmfamily]{Ligatures=TeX,Scale=1}
\fi
\usepackage{lmodern}
\ifPDFTeX\else
\fi
\IfFileExists{upquote.sty}{\usepackage{upquote}}{}
\IfFileExists{microtype.sty}{
  \usepackage[]{microtype}
  \UseMicrotypeSet[protrusion]{basicmath} 
}{}
\usepackage{xcolor}
\usepackage[margin=1in]{geometry}
\usepackage{graphicx}
\makeatletter
\def\maxwidth{\ifdim\Gin@nat@width>\linewidth\linewidth\else\Gin@nat@width\fi}
\def\maxheight{\ifdim\Gin@nat@height>\textheight\textheight\else\Gin@nat@height\fi}
\makeatother
\setkeys{Gin}{width=\maxwidth,height=\maxheight,keepaspectratio}
\makeatletter
\def\fps@figure{htbp}
\makeatother
\providecommand{\tightlist}{%
  \setlength{\itemsep}{0pt}\setlength{\parskip}{0pt}}
\NewDocumentCommand\citeproctext{}{}
\NewDocumentCommand\citeproc{mm}{%
  \begingroup\def\citeproctext{#2}\cite{#1}\endgroup}
\makeatletter
 \let\@cite@ofmt\@firstofone
 \def\@biblabel#1{}
 \def\@cite#1#2{{#1\if@tempswa , #2\fi}}
\makeatother
\newlength{\cslhangindent}
\newlength{\csllabelwidth}
\newenvironment{CSLReferences}[2] 
 {\begin{list}{}{%
  \setlength{\itemindent}{0pt}
  \setlength{\leftmargin}{0pt}
  \setlength{\parsep}{0pt}
  \ifodd #1
   \setlength{\leftmargin}{\cslhangindent}
   \setlength{\itemindent}{-1\cslhangindent}
  \fi
  \setlength{\itemsep}{#2\baselineskip}}}
 {\end{list}}
\usepackage{calc}

\usepackage{algorithm}
\usepackage{algorithmic}
\usepackage{amsmath}
\usepackage{amsthm}
\usepackage{amsmath}
\usepackage{amssymb}
\usepackage{bm}
\usepackage{booktabs}
\usepackage{braket}
\usepackage{caption}
\usepackage{float}
\usepackage[bottom,flushmargin,hang,multiple]{footmisc}
\usepackage{graphicx}
\usepackage{hyperref}
\usepackage{libertine}
\usepackage[libertine]{newtxmath}
\usepackage[mathcal]{euscript}
\usepackage[frozencache=true, cachedir=minted-cache]{minted}
\usepackage{multirow}
\usepackage{quoting}
\usepackage{setspace}
\usepackage{soul}
\usepackage{subcaption}
\usepackage{tabularx}
\usepackage{titlesec}
\usepackage[utf8]{inputenc}
\ifLuaTeX
  \usepackage{selnolig}  
\fi
\usepackage{bookmark}
\IfFileExists{xurl.sty}{\usepackage{xurl}}{} 
\hypersetup{
  pdftitle={Charting the Parrot's Song: A Maximum Mean Discrepancy Approach to Measuring AI Novelty, Originality, and Distinctiveness},
  colorlinks=true,
  linkcolor={Maroon},
  filecolor={Maroon},
  citecolor={Blue},
  urlcolor={blue},
  pdfcreator={LaTeX via pandoc}}

\title{Charting the Parrot's Song: A Maximum Mean Discrepancy Approach
to Measuring AI Novelty, Originality, and Distinctiveness}
\author{}
\date{\vspace{-2.5em}}

\begin{document}
\maketitle

\quotingsetup{font={itshape}, leftmargin=2em, rightmargin=2em, vskip=1ex}

\newtheorem{proposition}{Proposition} 

\vspace*{\fill}

\begin{center}
Anirban Mukherjee

Hannah Hanwen Chang

\bigskip

11 April, 2025
\end{center}

\vspace*{\fill}

\noindent \hrulefill

\noindent Anirban Mukherjee
(\href{mailto:anirban@avyayamholdings.com}{\nolinkurl{anirban@avyayamholdings.com}})
is Principal at Avyayam Holdings. Hannah H. Chang
(\href{mailto:hannahchang@smu.edu.sg}{\nolinkurl{hannahchang@smu.edu.sg}};
corresponding author) is Associate Professor of Marketing at the Lee
Kong Chian School of Business, Singapore Management University. This
research was supported by the Ministry of Education (MOE), Singapore,
under its Academic Research Fund (AcRF) Tier 2 Grant,
No.~MOE-T2EP40221-0008.

\newpage

\begin{center} 
\noindent \textbf{Abstract}
\end{center}

\noindent Current intellectual property frameworks struggle to evaluate
the novelty of AI-generated content, relying on subjective assessments
ill-suited for comparing effectively infinite AI outputs against prior
art. This paper introduces a robust, quantitative methodology grounded
in Maximum Mean Discrepancy (MMD) to measure distributional differences
between generative processes. By comparing entire output distributions
rather than conducting pairwise similarity checks, our approach directly
contrasts creative processes---overcoming the computational challenges
inherent in evaluating AI outputs against unbounded prior art corpora.
Through experiments combining kernel mean embeddings with
domain-specific machine learning representations (LeNet-5 for MNIST
digits, CLIP for art), we demonstrate exceptional sensitivity: our
method distinguishes MNIST digit classes with 95\% confidence using just
5--6 samples and differentiates AI-generated art from human art in the
AI-ArtBench dataset (n=400 per category; p\textless0.0001) using as few
as 7-10 samples per distribution despite human evaluators' limited
discrimination ability (58\% accuracy). These findings challenge the
``stochastic parrot'' hypothesis by providing empirical evidence that AI
systems produce outputs from semantically distinct distributions rather
than merely replicating training data. Our approach bridges technical
capabilities with legal doctrine, offering a pathway to modernize
originality assessments while preserving intellectual property law's
core objectives. This research provides courts and policymakers with a
computationally efficient, legally relevant tool to quantify AI
novelty---a critical advancement as AI blurs traditional authorship and
inventorship boundaries.

\begin{center}\rule{0.5\linewidth}{0.5pt}\end{center}

\noindent Keywords: Novelty, Originality, Distinctiveness, Artificial
Intelligence, Copyright, Patent, Intellectual Property Law.

\smallskip

\newpage

\singlespacing
\setcounter{tocdepth}{4}
\tableofcontents

\newpage

\doublespacing

\section{Introduction}\label{introduction}

\emph{\noindent "Because computers today, and for proximate tomorrows, cannot themselves formulate creative plans or `conceptions' to inform their execution of expressive works, they lack the initiative that characterizes human authorship. The computer scientist who succeeds at the task of `reduc[ing] [creativity] to logic' does not generate new `machine' creativity—she instead builds a set of instructions to codify and simulate `substantive aspect[s] of human [creative] genius,' and then commands a computer to faithfully follow those instructions. Even the most sophisticated generative machines proceed through processes designed entirely by the humans who program them, and are therefore closer to amanuenses than to true `authors.'"}

\noindent \hfill --- Ginsburg and Budiardjo
(\citeproc{ref-ginsburg2019authors}{2019}), p.~349.

\emph{\noindent "Notwithstanding its age and the technological advances that have occurred since its utterance, Lovelace’s critique remains credible. Even though today’s computers are exponentially more powerful than their early ancestors in terms of memory and processing, they still rely on humans in the first instance to dictate the rules according to which they perform. Like the photographer standing behind the camera, an intelligent programmer or team of programmers stands behind every artificially intelligent machine. People create the rules, and machines obediently follow them—doing, in Lovelace’s words, only whatever we order them to perform, and nothing more."}

\noindent \hfill --- Bridy (\citeproc{ref-bridy2012coding}{2012}),
p.~10.

\emph{\noindent "Use of texts to train LLaMA to statistically model language and generate original expression is transformative by nature and quintessential fair use—much like Google’s wholesale copying of books to create an internet search tool was found to be fair use in Authors Guild v. Google, Inc., 804 F.3d 202 (2d Cir. 2015)."}

\noindent \hfill --- \emph{R. Kadrey, S. Silverman, \& C. Golden v. Meta
Platforms, Inc.}, No.~3:23-cv-03417-VC.

\smallskip

The concepts of novelty, originality, and distinctiveness serve as
domain-specific criteria across various forms of intellectual property
(IP) law, each providing a framework for assessing how new creations
relate to existing knowledge. Patent law requires inventions to be
``novel'' and ``non-obvious'' compared to prior art.\footnote{These
  requirements for patentability are codified in Title 35 of the U.S.
  Code, primarily in 35 U.S.C. § 102 (novelty) and § 103
  (non-obviousness).} Trademark law mandates that marks exhibit
``distinctiveness,'' meaning they must sufficiently differentiate the
associated goods or services in the marketplace.\footnote{Trademark
  distinctiveness is governed by the Lanham Act, 15 U.S.C. § 1051
  \emph{et seq.}, and is often analyzed along a spectrum from generic to
  arbitrary or fanciful, potentially including acquired distinctiveness
  (secondary meaning).} Copyright law requires ``originality,'' meaning
independent creation with at least a minimal degree of
creativity.\footnote{Copyright protection under 17 U.S.C. § 102(a)
  extends to ``original works of authorship,'' a standard requiring both
  independent creation and a minimal level of creativity.}

While these concepts operate differently within their respective
domains, they share a common function: measuring the degree to which new
creations depart from prior works. Foundational cases---such as
\emph{Graham v. John Deere Co.} (383 U.S. 1, 1966) for patent novelty,
\emph{Abercrombie \& Fitch Co.~v. Hunting World, Inc.} (537 F.2d 4, 2d
Cir. 1976) for trademark distinctiveness, and \emph{Feist Publications,
Inc.~v. Rural Telephone Service Co.} (499 U.S. 340, 1991) for copyright
originality---along with leading treatises (e.g.,
\citeproc{ref-chisum2022patents}{Chisum 2022} on Patents;
\citeproc{ref-mccarthy2025trademarks}{McCarthy 2025} on Trademarks;
\citeproc{ref-nimmer2023copyright}{Nimmer and Nimmer 2023} on
Copyright), underscore the importance of effectively measuring the
relationship between new creations and existing works. In patent law,
this involves comparing new inventions to the existing body of knowledge
(prior art); in copyright and trademark law, it involves comparing
independent creative works to existing works. Across these domains,
questions of comparative distinctiveness---broadly understood as the
measurable differentiation between a new creation and existing
knowledge, or between two independent works---often form the crux of
legal disputes.

This established principle of assessing comparative distinctiveness,
however, faces unprecedented challenges due to recent advances in
artificial intelligence (AI). This is particularly evident in ongoing
debates surrounding AI authorship. Currently, the U.S. Copyright Office,
along with many international jurisdictions, maintains that works
generated solely by AI---without human authorship---are not eligible for
copyright protection (\citeproc{ref-guadamuz2016monkey}{Guadamuz
2016}).\footnote{This position aligns with the traditional view of such
  systems as mere tools or ``amanuenses'' incapable of independent
  creation. \emph{See} U.S. Copyright Office, \emph{Compendium of U.S.
  Copyright Office Practices} § 313.2 (3d ed.~2021). The Office
  reiterated this stance in recent guidance, emphasizing that copyright
  protection requires works to be the product of human authorship and
  refusing registration for works where AI contributions are not the
  result of human creative control or where the human contribution lacks
  sufficient originality. \emph{See} U.S. Copyright Office,
  \emph{Copyright Registration Guidance: Works Containing Material
  Generated by Artificial Intelligence}, 88 Fed. Reg. 16190 (Mar.~16,
  2023).} This stance was notably applied in the case of the AI-assisted
comic \emph{Zarya of the Dawn}, where registration for the work as a
whole was refused because the human user's text prompts were deemed
insufficient to constitute the necessary creative input for authorship
of the AI-generated images.\footnote{\emph{See} U.S. Copyright Office,
  Letter re: Zarya of the Dawn (Registration \# VAu001480196) (Feb.~21,
  2023) (concluding that the AI-generated images were not products of
  human authorship, while granting protection to the text and the
  selection/arrangement of elements authored by Kristina Kashtanova).}
\footnote{This stance contrasts with approaches in some other
  jurisdictions; for instance, Chinese courts have reached differing
  conclusions, sometimes granting copyright protection based on the
  human team's role in selecting data and parameters that guided the
  AI's output, effectively recognizing the human orchestration of the
  generative process. For instance, compare \emph{Shenzhen Tencent
  Computer System Co., Ltd.~v. Shanghai Yingxun Technology Co., Ltd.},
  {[}2019{]} Yue 0305 Min Chu 14010 (Shenzhen Nanshan Dist. People's Ct.
  Dec.~24, 2019) (granting protection based on human selection and
  arrangement) with \emph{Beijing Film Law Firm v. Beijing Baidu Netcom
  Science \& Technology Co., Ltd.}, {[}2018{]} Jing 0491 Min Chu No.~239
  (Beijing Internet Ct. Apr.~25, 2019) (denying protection, requiring
  natural person creation). For discussion, see Wan and Lu
  (\citeproc{ref-wan2021copyright}{2021}).} Although legal debates and
lawsuits related to AI-generated content continue to evolve across
intellectual property domains\footnote{E.g., \emph{Thaler v.
  Perlmutter}, No.~22-1564 (BAH) (D.D.C. Aug.~18, 2023) (denying patent
  inventorship to AI), and European Patent Office (EPO) Legal Board of
  Appeal, Case J 8/20 (Dec.~21, 2021) (same).}, the broad consensus
remains that AI systems, in their current form, cannot satisfy the
traditional requirements of human authorship or inventorship.\footnote{\emph{Also
  see}, Sun (\citeproc{ref-sun2021redesigning}{2021}).}

This perspective aligns with the longstanding view---tracing back to Ada
Lovelace---that without human authorship, a creative work cannot meet
the threshold of originality required by copyright law
(\citeproc{ref-bridy2012coding}{Bridy 2012};
\citeproc{ref-schafer2015fourth}{Schafer et al. 2015}). As Ginsburg and
Budiardjo (\citeproc{ref-ginsburg2019authors}{2019}) forcefully state,
even the most sophisticated AI systems ``lack the initiative that
characterizes human authorship'' and are ``closer to amanuenses than to
true `authors'\,'' (p.~349). They conceive authorship as resting on two
pillars: a mental step (the conception of a work) and a physical step
(the execution of a work). They exclude AI from the former as current AI
systems lack genuine cognitive agency or motivation, and from the latter
because they view AI outputs as strictly determined by human-programmed
instructions, making AI systems closer to amanuenses than authors. Thus,
they conclude, AI systems fail to achieve originality in either
conception or execution.

However, there are grounds to expect AI outputs to be novel. Because AI
systems necessarily combine and interpolate between their training
points, their outputs are almost always structurally
distinct.\footnote{In a fundamental mathematical sense, almost
  everything modern generative AI systems produce (with probability
  approaching one) is novel, as these systems operate based on
  probabilistic relationships among elements (e.g., words, pixels) and
  concepts, rather than by retrieving pre-existing content. For
  instance, in text generation, large language models interpolate
  between words, where all inputs and prior outputs define the
  probabilities used to sample the next word. Similarly, diffusion and
  flow models map points from a high-dimensional continuous sample space
  to images, such that prior training examples correspond only to
  discrete points within that space.} \footnote{\emph{See} Degli Esposti
  et al. (\citeproc{ref-degli2020use}{2020}) for examples of AI systems
  whose ``creativity'' is not based on pre-existing works.} As each
output element is recursively fed back into the model, the resulting
outputs naturally diverge from their original sources, occasionally
losing their original meaning or even creating entirely new
``facts''---a phenomenon known as hallucination or confabulation
(\citeproc{ref-Ji2023}{Ji et al. 2023};
\citeproc{ref-mukherjee2023managing}{Mukherjee and Chang 2023}). Indeed,
this perspective is central to Meta's defense in \emph{R. Kadrey, S.
Silverman, \& C. Golden v. Meta Platforms, Inc.}, No.~3:23-cv-03417-VC:
if an AI's training inputs serve merely as points for interpolation,
then its outputs may often be \emph{functionally}
transformative\footnote{The concept of transformative use, where a new
  work alters the original with new expression, meaning, or message, is
  central to fair use analysis in copyright law. \emph{See, e.g.,
  Campbell v. Acuff-Rose Music, Inc.}, 510 U.S. 569 (1994); \emph{Cariou
  v. Prince}, 714 F.3d 694 (2d Cir. 2013). Applying this concept to AI
  outputs trained on copyrighted data is a key issue in ongoing
  litigation.} rather than direct reflections of its training data, and
therefore not necessarily \emph{functionally} derivative.\footnote{The
  term ``functionally derivative'' is used here to describe AI outputs
  that operationally resemble derivative works as defined in 17 U.S.C. §
  101, which are works ``based upon one or more preexisting works''
  through recasting, transformation, or adaptation. However, this
  characterization does not imply legal status. Under current U.S.
  copyright law, derivative works require human authorship and
  intentional adaptation or transformation of preexisting works (17
  U.S.C. § 106(2)). AI systems, lacking human authorship and the
  requisite intent (\emph{mens rea}), cannot legally create derivative
  works. The U.S. Copyright Office explicitly maintains that copyright
  protection requires human authorship. \emph{See} U.S. Copyright
  Office, \emph{Compendium of U.S. Copyright Office Practices} § 313.2
  (3d ed.~2021); \emph{see also Thaler v. Perlmutter}, No.~22-1564 (BAH)
  (D.D.C. Aug.~18, 2023). Thus, the term ``functionally derivative''
  emphasizes operational similarity without conferring legal authorship
  or infringement capacity upon the AI itself.} Evaluating such claims
requires robust methods capable of assessing the \emph{degree} of
distinctiveness between an AI's output distribution and the distribution
of prior art.

\subsection{Assessing Novelty, Originality, and
Distinctiveness}\label{assessing-novelty-originality-and-distinctiveness}

While the lack of genuine cognitive agency (conception) in AI remains
largely undisputed at present, we argue that the lack of originality in
AI \emph{execution} is often more assumed than empirically measured---in
part due to the absence of a suitable empirical measure, a gap this
paper seeks to address. This challenge is particularly acute in legal
contexts, where human contribution is paramount. For instance, recent
guidance from the U.S. Patent and Trademark Office (USPTO) on
AI-assisted inventions reaffirms that only natural persons can be
inventors, but clarifies that AI assistance does not preclude
patentability if a human provides a ``significant
contribution.''\footnote{See \emph{Inventorship Guidance for AI-Assisted
  Inventions}, 89 Fed. Reg. 10043 (Feb.~13, 2024). The guidance
  emphasizes that the inventorship analysis must focus on human
  contributions and applies the \emph{Pannu} factors (\emph{Pannu v.
  Iolab Corp.}, 155 F.3d 1344, 1351 (Fed. Cir. 1998)) to determine if a
  human's contribution to the conception of the AI-assisted invention
  was significant.} This legal framework, while necessary for
determining inventorship, relies on assessing factors such as the
human's contribution to conception and whether it was ``not
insignificant in quality.'' Such assessments often involve qualitative
judgments about the \emph{human's actions} rather than a direct
quantitative measure of the \emph{output's distinctiveness}.
Furthermore, traditional qualitative assessments of novelty across IP
domains rely on subjective judgments about a work's originality,
significance, and impact. Such judgments can vary widely, encompassing
everything from incremental improvements to groundbreaking innovations,
leaving ample room for selective interpretation and reinforcing existing
biases regarding AI's capacity for genuine innovation.\footnote{It is
  noteworthy that trademark law diverges from copyright and patent law
  in this regard; there is currently no specific U.S. statute or
  regulation requiring human ``creation'' for a trademark, as the focus
  remains on the mark's use by a legal person to identify source.
  However, the capacity of generative AI to easily create numerous
  potential marks raises significant practical concerns. This ease of
  generation risks an oversaturation of the trademark landscape,
  potentially diluting the ability of any mark to serve its essential
  function as a unique source identifier. Therefore, evaluating the
  differentiation between a mark or a set of marks generated by AI and
  the vast field of existing marks (human or otherwise) becomes an
  increasingly complex and vital task. This situation underscores the
  critical need for robust methods to assess comparative
  distinctiveness, as explored in this paper.}

Moreover, traditional quantitative metrics of novelty, originality, and
distinctiveness in natural language processing---such as cosine
similarity---typically rely on pairwise comparisons, which are direct
evaluations between individual works, rather than assessing differences
between the underlying generative (creative) processes
(\citeproc{ref-vsavelka2022legal}{Šavelka and Ashley 2022}). For
instance, in visual art, these quantitative measures might compare
individual paintings---one painting by an artist against another
painting by a different artist---to gauge novelty. However, they cannot
directly evaluate the novelty of one painter's \emph{entire} creative
process relative to another's. As a result, attempts to capture
process-to-process novelty comparisons using existing methods inevitably
depend either on qualitative judgments or on \emph{ad hoc} aggregations
of pairwise distance metrics (such as the mean or maximum of the
pairwise similarities between the outputs of two artists). This approach
lacks a principled and consistent quantitative basis.

Measuring the \emph{difference} between the generative process of an AI
and the human creative processes underlying prior art\footnote{Here and
  subsequently, ``prior art''---while technically a patent law term---is
  used more broadly to denote the relevant collection of existing
  domain-specific items (e.g., prior inventions, existing copyrighted
  works, registered trademarks). This generalized usage facilitates a
  consistent discussion of comparing new creations to an existing corpus
  across different IP fields.} is particularly essential for several
reasons. It has always been impractical to comprehensively collect and
analyze the entirety of human-generated prior art---a longstanding
challenge even in traditional assessments of novelty. AI introduces an
additional complication: because the generative capacity of AI is
effectively infinite, comparing an AI's outputs to prior art requires an
infinite number of comparisons. Furthermore, as AI-generated outputs
themselves become part of prior art, both the body of prior art and the
set of AI outputs expand indefinitely, rendering traditional pairwise
comparisons intractable. Moreover, as these sets expand, even genuinely
innovative AI outputs will increasingly coincide with prior human or AI
creations purely by chance, misleadingly suggesting that the AI merely
replicates existing content (\citeproc{ref-villasenor2023ten}{Villasenor
2023}). These issues further limit the utility of traditional
quantitative metrics.

\subsection{Maximum Mean Discrepancy
(MMD)}\label{maximum-mean-discrepancy-mmd}

Consistent with the need to measure novelty and aligned with calls for a
realistic understanding of AI's current capabilities and limitations
(\citeproc{ref-surden2018artificial}{Surden 2018}), we propose using
Maximum Mean Discrepancy (MMD) as the basis for a quantitative measure
of novelty. MMD is a kernel-based statistical approach designed to
measure the distance between two probability distributions---not by
comparing individual samples from these distributions, but by examining
their collective properties.\footnote{Rather than making direct pairwise
  comparisons between individual samples, MMD evaluates whether samples
  drawn from one distribution can, as a group, be reliably distinguished
  from samples drawn from another distribution. This approach
  capitalizes on systematic differences across the entire sample space,
  rather than idiosyncratic points of comparison.}

This approach is particularly valuable in the context of AI-generated
content, where individually comparing every possible AI output against
the vast body of existing prior art is impractical. Instead, MMD allows
us to ask a simpler, more practical question: Do the outputs from an AI
system, viewed collectively, tend to resemble the kinds of works already
produced by humans, or do they differ in meaningful ways? If an AI
system merely replicates or closely imitates its training data (the
prior art), its outputs, taken together, will appear very similar to
that data. Conversely, if the AI system produces genuinely novel
outputs, its outputs, taken together, will differ significantly. By
focusing on these distribution-level differences rather than individual
comparisons, MMD provides a robust and practical way to assess whether
an AI's creative process is meaningfully distinct from human creative
processes.

This shift in approach offers several advantages. First, by measuring
novelty holistically at the process level, we address the concern that
even genuinely innovative AI systems might occasionally produce outputs
that \emph{coincidentally} resemble prior art, simply due to the
vastness of both sets; by evaluating the \emph{overall tendencies} of
generative processes rather than individual outputs, our method
accommodates similarities (and differences) arising purely by chance.
Second, although we aim to determine whether a potentially infinite set
of works (e.g., AI outputs) differs from another potentially infinite
set (e.g., prior art), our method must remain practical and estimable
using only finite samples from each distribution. MMD is particularly
well-suited to this scenario, as it provides a statistically robust
approach for estimating distribution-level differences from relatively
small sample sizes.\footnote{Our empirical work shows that as few as 5
  samples from each distribution may suffice to ensure robust inference.}
Consequently, our method does not require exhaustive knowledge of all
possible AI outputs or a complete catalog of prior art.

To ensure our metric captures \emph{semantic} information, we leverage
machine learning embeddings---mathematical functions that map
unstructured data, such as text or images, into high-dimensional vector
spaces (\citeproc{ref-chalkidis2019deep}{Chalkidis and Kampas 2019}).
Similar embedding-based approaches have been successfully applied to
quantify distinctiveness in intellectual property contexts, particularly
in assessing trademark registrability
(\citeproc{ref-adarsh2024automating}{Adarsh et al. 2024}). Our work
extends these techniques to the novel context of evaluating the
distinctiveness of creative outputs across intellectual property
domains. These embeddings preserve semantic relationships by positioning
semantically similar items closer together and dissimilar items farther
apart, thereby capturing underlying meaning and context. By combining
embeddings with MMD, we measure the \emph{semantic} distance between two
creative processes, providing a robust and meaningful quantitative
assessment of their (dis)similarity.

\subsection{Empirical Validation}\label{empirical-validation}

We validate our methodology across two increasingly complex visual
domains. First, we establish the statistical robustness of our method
using the MNIST dataset of handwritten digits, where we have clear
ground truth regarding distributional differences. This controlled
experiment demonstrates our method's ability to distinguish between
distributions even with limited sample sizes, quantify degrees of
difference between similar and dissimilar distributions, and establish
appropriate statistical confidence thresholds. By embedding digit images
using a convolutional neural network (LeNet-5) and applying our MMD
framework, we systematically evaluate both the sensitivity and
specificity of our approach.

Second, we extend our validation to a more challenging real-world domain
by analyzing the AI-ArtBench dataset
(\citeproc{ref-silva2024artbrain}{Silva et al. 2024}), which contains
185,015 artistic images across ten art styles---including 60,000
human-created artworks and 125,015 AI-generated images produced by two
different generative models (Latent Diffusion and Standard Diffusion).
This dataset is particularly valuable for our purposes, as recent
research demonstrates that humans can identify AI-generated images with
only approximately 58\% accuracy, highlighting the increasingly blurred
line between human and AI creativity in the visual arts. By leveraging
CLIP embeddings to capture semantic and stylistic elements of the
artwork, we test whether our MMD-based approach can detect statistically
significant differences between human-created and AI-generated
distributions---and between different AI generation techniques---that
might elude human perception. This application directly addresses
whether AI-generated art remains statistically distinguishable from
human-created art, even as visual differences become increasingly
subtle.

\subsection{Contributions and
Organization}\label{contributions-and-organization}

First and foremost, we contribute a novel methodological framework with
significant implications for IP law. Our methodology shifts the focus
from comparing individual outputs to assessing the distinctiveness of
the underlying generative processes. By combining kernel mean embeddings
(KME), MMD, and domain-specific machine learning embeddings, we directly
address fundamental limitations of traditional legal assessments: the
impossibility of exhaustive pairwise comparisons between effectively
infinite sets of AI outputs and prior art, the complexities arising from
combining pairwise similarity metrics, and the inherent subjectivity of
qualitative comparisons. In contrast, we offer a statistically robust
metric to determine whether an AI's creative process is meaningfully
different from the processes that generated existing works.

Our approach is designed to be practicable. Unlike AI detection systems
that require extensive training data and model-specific tuning (e.g.,
\citeproc{ref-mukherjee2024safeguarding}{Mukherjee 2024}), our method
requires no training data and operates effectively with limited samples
(as few as five samples from each distribution). This data efficiency is
crucial for contexts where comprehensive datasets are often unavailable
or short, such as evaluating the novelty of AI-generated works against a
single artist's portfolio or assessing trademark distinctiveness in
specialized markets. This practicality makes our approach immediately
applicable in real-world legal settings, providing courts and
policymakers with a principled, quantitative tool for assessing AI
novelty that aligns with established statistical methods.

Moreover, we provide statistically significant evidence that
AI-generated outputs can be distinct from prior art. By demonstrating
that AI systems can exhibit a measurable degree of novelty at the
process level, we inform ongoing legal debates (e.g., \emph{Kadrey v.
Meta}, 2023) that center on whether AI-generated content represents
meaningful creative contributions or mere recombinations of existing
works.

Central to these debates is the argument colloquially known as the
``stochastic parrot'' critique. This view holds that AI systems merely
replicate learned patterns with superficial variations, lacking genuine
understanding or creative intent
(\citeproc{ref-bender2021dangers}{Bender et al. 2021}). Consequently, AI
outputs are seen as inherently \emph{functionally}
derivative,\footnote{A ``derivative work'' under 17 U.S.C. § 101
  involves recasting or adapting preexisting works. While AI outputs may
  adapt source material in ways that resemble derivative works, AI
  systems legally cannot be authors (17 U.S.C. § 106(2)) nor possess the
  requisite intent (\emph{mens rea}). The term ``functionally
  derivative'' denotes this operational resemblance without implying a
  legal status.} substantially based on or adapted from prior works,
reflecting statistical correlations in their training data rather than
original thought.\footnote{Much of the current legal debate, including
  lawsuits against AI developers, centers on whether AI outputs are
  substantially similar to, and therefore infringing derivatives of, the
  copyrighted works within their vast training datasets. \emph{See},
  e.g., \emph{Authors Guild et al.~v. OpenAI Inc.}, No.~1:23-cv-08292
  (S.D.N.Y. 2023); \emph{Andersen et al.~v. Stability AI Ltd.},
  No.~3:23-cv-00201 (N.D. Cal. 2023). Although other factors like the
  idea/expression dichotomy and normative questions are relevant
  (\citeproc{ref-grimmelmann2015there}{Grimmelmann 2015};
  \citeproc{ref-lemley2023generative}{Lemley 2023}), the stochastic
  parrot critique underpins arguments against AI originality
  (\citeproc{ref-marcus2019rebooting}{Marcus and Davis 2019}).}

Prior empirical findings on the novelty of AI-generated content are
profoundly split. On the one hand, research documents AI systems
memorizing and reproducing their training data
(\citeproc{ref-Copyleaks_2024}{Copyleaks 2024}). Studies employing
methods such as verbatim text matching have revealed substantial copying
(\citeproc{ref-lee2022deduplicating}{Lee et al. 2022};
\citeproc{ref-chang2023speak}{Chang et al. 2023};
\citeproc{ref-carlini2023extracting}{Nasr et al. 2023}), with larger
models showing a greater propensity for memorization
(\citeproc{ref-diakopoulos2023memorized}{Diakopoulos 2023}). These
findings lend support to the stochastic parrot hypothesis and feature
prominently in legal arguments concerning substantial similarity and
infringement.\footnote{E.g., \emph{Sarah Andersen et al.~v. Stability AI
  Ltd., Midjourney Inc., and DeviantArt Inc.}, No.~3:23-cv-00201-WHO
  (N.D. Cal. 2023); \emph{Authors Guild et al.~v. OpenAI Inc.},
  No.~1:23-cv-08292 (SHS) (S.D.N.Y. 2023); and \emph{Getty Images (US),
  Inc.~v. Stability AI, Inc.}, No.~1:23-cv-00135 (D. Del. 2023).} On the
other hand, a growing body of evidence, often relying on semantic
analysis and human evaluations, challenges the portrayal of AI systems
as mere mimics. For instance, analyses such as RAVEN suggest
AI-generated text can achieve high structural or thematic novelty
despite lower local novelty (\citeproc{ref-mccoy2023much}{McCoy et al.
2023}). Other work shows AI achieving human-like systematic
generalization (\citeproc{ref-lake2023human}{Lake and Baroni 2023}) or
performing well on scholarly novelty benchmarks
(\citeproc{ref-lin2024evaluating}{Lin et al. 2024}), suggesting AI can
generate outputs that meaningfully diverge from training data.

While much of this empirical debate has centered on text, our research
addresses the stochastic parrot narrative within the visual domain using
a distributional perspective. We demonstrate that AI-generated artworks
are consistently distinguishable from human-created works at the
distributional level, even when human evaluators struggle to visually
discriminate between them. Notably, these differences emerge even at
small sample sizes (as few as 7), suggesting the divergence is
fundamental. Our methodology detects these systematic differences while
addressing limitations in previous research: unlike memorization studies
focusing on exact matches, our approach captures distributional novelty;
unlike semantic evaluations potentially relying on subjective judgments,
our framework provides an objective, quantifiable metric. By measuring
novelty at the process level, we offer empirical evidence that, at least
in the visual domains studied, AI systems do more than merely recombine
elements---they generate outputs from a statistically distinct
distribution.

The remainder of our paper is organized as follows. Section 2 provides a
detailed explanation of our MMD-based methodology. Section 3 presents
validation results from the MNIST and AI-ArtBench applications,
demonstrating the performance of the methodology under controlled
conditions and with real-world visual data, where human perception
struggles to distinguish between AI and human origins. The final section
discusses the implications of our findings, addresses limitations,
outlines directions for future research, and concludes.

\section{Method Development}\label{method-development}

The core question we address is whether one body of content (e.g.,
AI-generated outputs) is statistically distinguishable from another
(e.g., prior art)---that is, whether the two bodies of content are
distinct with very high probability. Our approach is based on a
straightforward intuition: consider the probability of a particular
document (e.g., an image) arising from two distinct generative
processes. If an AI system merely reproduces what it has previously
encountered, its generative process will mirror that of prior art; the
output would be equally likely to emerge from the AI as from the human
creative processes underlying prior art. Conversely, if the AI is
genuinely innovative, its outputs will differ \emph{systematically} from
prior art. Certain documents will have different probabilities of
arising from the AI than from prior art, indicating that the AI is not
simply replicating existing content. In other words, true novelty
manifests at the \emph{distributional} level.

To this end, we propose a statistical framework based on KMEs (for
detailed technical derivations and properties, see
\citeproc{ref-gretton2012kernel}{Gretton et al. 2012};
\citeproc{ref-muandet2017kernel}{Muandet et al. 2017})\footnote{The
  mathematics underlying KMEs is complex. Here, we provide a discussion
  tailored to our specific use; additional details can be found in the
  referenced works, with an exhaustive presentation in Berlinet and
  Thomas-Agnan (\citeproc{ref-berlinet2011reproducing}{2011}).}, MMD,
and machine learning embeddings. Our methodology integrates two
complementary strands of research on embeddings---mappings that
transform mathematical objects (e.g., text or images) into a new space
while preserving key relationships. One strand defines abstract notions
of embeddings and establishes formal properties useful for theoretical
analysis (\citeproc{ref-sriperumbudur2010hilbert}{Sriperumbudur et al.
2010}). The other develops practical algorithms, which we term
\emph{machine learning embeddings}, for discovering effective embeddings
in various domains, such as text and images
(\citeproc{ref-mikolov2013efficient}{Mikolov et al. 2013}). We combine
these approaches to create a unified framework for distinctiveness and
novelty analysis.

Specifically, we first employ a machine learning embedding algorithm to
represent both prior art and AI-generated outputs (which may be
non-numerical, such as images) in a vector space. These machine learning
embeddings represent complex data as points, where distances between
points reflect semantic relationships among the original data items.
They enable numerical analysis of the non-numerical data, providing a
natural measure of distance derived from the semantic content of the
embedded objects (\citeproc{ref-stammbach2021docscan}{Stammbach and Ash
2021}). We then use these representations to construct KMEs of the
distributions of \emph{both} prior art and AI-generated outputs. This
compositional mapping (from the non-numerical data to the numerical
vector space, and then via the kernel's feature map into a reproducing
kernel Hilbert space (RKHS) of functions) allows us to define an MMD---a
type of integral probability metric (IPM)---within the RKHS. This
approach yields a metric for hypothesis testing to determine whether two
creative processes are statistically distinguishable.

The IPM provides a principled way to measure the distance between two
probability distributions. When applied to AI outputs and prior art
using the compositional feature map described above, the IPM in the
resulting RKHS corresponds to the distance between the underlying
generative processes, directly quantifying systemic novelty.

\subsection{Definitions and
Background}\label{definitions-and-background}

Let \(X = \{x_1, x_2, \dots, x_m\}\) be a sample of embedded
AI-generated outputs drawn from an unknown probability distribution
\(P\), and let \(Y = \{y_1, y_2, \dots, y_n\}\) be a sample of embedded
prior art outputs drawn from an unknown probability distribution \(Q\).
Our goal is to test the null hypothesis \(H_0: P = Q\) (the
distributions are identical) against the alternative hypothesis
\(H_1: P \neq Q\) (the distributions differ).

An RKHS \(\mathcal{H}\) is a Hilbert space of functions defined by a
positive definite kernel function
\(k: \mathcal{X} \times \mathcal{X} \rightarrow \mathbb{R}\), where
\(\mathcal{X}\) is the input space (e.g., the space of possible AI
outputs \(X\), or the space of prior art \(Y\)). As an RKHS is a type of
Hilbert space (i.e., a complete inner product space), it inherits all of
its properties.

What distinguishes an RKHS from other function spaces is its reproducing
property. For every function \(f\) in the RKHS and every point
\(x \in \mathcal{X}\), the value of the function at that point,
\(f(x)\), can be reproduced by the inner product of \(f\) with the
kernel evaluation function, which is the kernel function centered at
\(x\), \(k(\cdot, x)\): \[
f(x) = \langle f, k(\cdot, x) \rangle.
\]\\
The kernel function provides a way to ``probe'' the function \(f\) at
any point \(x\) through the inner product. It allows us to represent
high-dimensional or even infinite-dimensional feature spaces implicitly,
which is a cornerstone of kernel methods in machine learning
(\citeproc{ref-shawe2004kernel}{Shawe-Taylor and Cristianini 2004};
\citeproc{ref-steinwart2008support}{Steinwart and Christmann 2008}).

A KME leverages the machinery of RKHS to embed probability distributions
into a Hilbert space. Specifically, given a probability distribution
\(P\) over a domain \(\mathcal{X}\), and a reproducing kernel \(k\) that
induces the RKHS \(\mathcal{H}\), the KME of \(P\) into \(\mathcal{H}\)
is the expected value of the kernel evaluation function (associated with
\(k\)) over \(P\). Mathematically, the embedding \(\mu_P\) of \(P\) is
given by: \[
\mu_P = \mathbb{E}_{X \sim P}[k(X, \cdot)] = \int_{\mathcal{X}} k(x, \cdot) \, dP(x),
\] where \(k(X, \cdot)\) represents the kernel evaluation function, a
function in the RKHS defined by fixing one argument of the kernel:
\(k(x, \cdot): y \mapsto k(x, y)\). The integral
\(\int_{\mathcal{X}} k(x, \cdot) \, dP(x)\) is a Bochner integral.

A KME maps an entire probability distribution \(P\) to a single,
corresponding function in the RKHS \(\mathcal{H}\). If the kernel \(k\)
is \emph{characteristic}, then the mapping from distributions to their
KMEs is \emph{injective} (one-to-one). This means that, given a
characteristic kernel, for any two probability distributions \(P\) and
\(Q\) on \(\mathcal{X}\), if their KMEs are equal (\(\mu_P = \mu_Q\)),
then the distributions themselves must be equal (\(P = Q\)). Conversely,
if \(P \neq Q\), then \(\mu_P \neq \mu_Q\).

Intuitively, a characteristic kernel ensures that if two probability
distributions differ, their kernel mean embeddings will also differ.
This builds on the notion that a kernel function measures the similarity
between two points in the input space (\(\mathcal{X}\)); it follows that
the distance between the embeddings of two distributions in the RKHS
corresponds to the similarity of samples (in the input space) drawn from
these distributions.

MMD is a statistic that quantifies the distance between two probability
distributions, \(P\) and \(Q\), as the distance between their KMEs in
the RKHS (\citeproc{ref-gretton2012kernel}{Gretton et al. 2012}): \[
\text{MMD}^2(P, Q) = \| \mu_P - \mu_Q \|_{\mathcal{H}}^2,
\] where \(\|\cdot\|_{\mathcal{H}}\) denotes the norm in the RKHS.

More generally, an IPM between distributions \(P\) and \(Q\) is defined
as: \[
\text{IPM}(P, Q) = \sup_{f \in \mathcal{F}} \left| \int_{\mathcal{X}} f(x) \, dP(x) - \int_{\mathcal{X}} f(x) \, dQ(x) \right|,
\] where \(\mathcal{F}\) is a class of functions. Thus, MMD is a type of
IPM, where the class of functions \(\mathcal{F}\) is the unit ball in
the RKHS.

\subsection{Employing MMD to Measure
Novelty}\label{employing-mmd-to-measure-novelty}

Suppose both the AI's outputs and prior art are numerical data. Given
samples \(X\) and \(Y\) from distributions \(P\) and \(Q\),
respectively, we can use an \emph{unbiased} empirical estimator of
\(\text{MMD}^2\): \[
\widehat{\text{MMD}}^2_u(X, Y) = \frac{1}{m(m-1)} \sum_{i=1}^{m} \sum_{\substack{j=1 \\ j \neq i}}^{m} k(x_i, x_j) + \frac{1}{n(n-1)} \sum_{i=1}^{n} \sum_{\substack{j=1 \\ j \neq i}}^{n} k(y_i, y_j) - \frac{2}{mn} \sum_{i=1}^{m} \sum_{j=1}^{n} k(x_i, y_j).
\]

This estimator can be computed efficiently using the \emph{kernel
trick}, which avoids explicit computation of the feature maps
\(k(\cdot, x)\). Specifically, the components of this equation are
interpreted as follows:

\begin{itemize}
\tightlist
\item
  \(k(\cdot, \cdot)\) is the kernel function used within the RKHS.
\item
  \(x_i\) and \(x_j\) are samples drawn from distribution \(P\).
\item
  \(y_i\) and \(y_j\) are samples drawn from distribution \(Q\).
\item
  \(m\) and \(n\) are the sample sizes drawn from distributions \(P\)
  and \(Q\), respectively.
\item
  The first term,
  \(\frac{1}{m(m-1)} \sum_{i=1}^{m} \sum_{\substack{j=1 \\ j \neq i}}^{m} k(x_i, x_j)\),
  calculates the average of the kernel evaluations over all unique pairs
  of samples from \(P\).
\item
  The second term,
  \(\frac{1}{n(n-1)} \sum_{i=1}^{n} \sum_{\substack{j=1 \\ j \neq i}}^{n} k(y_i, y_j)\),
  calculates the average of the kernel evaluations over all unique pairs
  of samples from \(Q\).
\item
  The third term,
  \(-\frac{2}{mn} \sum_{i=1}^{m} \sum_{j=1}^{n} k(x_i, y_j)\), subtracts
  twice the average of the kernel evaluations between samples from \(P\)
  and samples from \(Q\).
\end{itemize}

In general, we would like to apply this framework to various data types
(e.g., text, images) and not just numerical data. Therefore, we propose
first mapping the raw data into a numerical vector space using a machine
learning embedding.

Let \(\phi_x: \mathcal{X} \rightarrow \mathcal{Z}\) represent this
embedding, where \(\mathcal{Z}\) is typically \(\mathbb{R}^d\), with
\(d\) being the dimensionality of the embedding space. The choice of
embedding depends on the specific data type (e.g., a text embedding for
text data, a convolutional neural network (CNN) embedding for images).
This embedding should capture relevant relationships between data points
(e.g., semantic similarity for text, visual similarity for images).

We propose the following compositional structure for the feature map: \[
\phi(x) = \phi_k(\phi_x(x)), \quad x \in \mathcal{X},
\] where:

\begin{itemize}
\tightlist
\item
  \(\phi_x(x)\) is the machine learning embedding of the raw data point
  \(x\).
\item
  \(\phi_k\) is the feature map defined implicitly by the choice of
  kernel \(k\) in the RKHS.
\end{itemize}

This compositional map \(\phi\) has the following interpretation:
\(\phi_x\) maps the non-numerical data into a numerical vector space
with sufficient dimensionality to distinguish between any two distinct
elements, and \(\phi_k\) is the feature map in an RKHS with sufficient
richness to ensure that a KME \(\mu_P\) accurately describes \(P\).

Crucially, the kernel defined by this compositional structure is
characteristic if the machine learning embedding is injective and if the
kernel in the RKHS is characteristic. Formally,
\(k(\phi_x(x_i), \phi_x(x_j))\) is characteristic if \(k\) is
characteristic on \(\mathcal{Z}\) and \(\phi_x\) is injective (see
Proposition \ref{prop:composed_kernel} below). This allows us to apply
our MMD framework to any data type for which a suitable embedding
\(\phi_x\) can be found. The MMD estimator is computed by evaluating the
kernel function on the embedded data, replacing \(k(x_i, x_j)\) with
\(k(\phi_x(x_i), \phi_x(x_j))\) in the formula above.

\begin{proposition} \label{prop:composed_kernel}
If \(\phi_x: \mathcal{X} \to \mathcal{Z}\) is injective and the kernel \(k\) is characteristic on \(\mathcal{Z}\), then the composed kernel \(k_{\phi}(x_i, x_j) = k(\phi_x(x_i), \phi_x(x_j))\) is characteristic on \(\mathcal{X}\).
\end{proposition}

\begin{proof}
Let \(P_x\) and \(Q_x\) be two probability measures on \(\mathcal{X}\). Define their pushforward measures \(P_z = \phi_{x\#}P_x\) and \(Q_z = \phi_{x\#}Q_x\) on \(\mathcal{Z}\), respectively. By definition, for any measurable set \(B \subseteq \mathcal{Z}\):
\[
P_z(B) = P_x(\phi_x^{-1}(B)), \quad Q_z(B) = Q_x(\phi_x^{-1}(B)).
\]

Since \(\phi_x\) is injective, it follows immediately that if \(P_x \neq Q_x\), then there exists a measurable set \(A \subseteq \mathcal{X}\) such that \(P_x(A) \neq Q_x(A)\). Letting \(B = \phi_x(A)\), we have:
\[
P_z(B) = P_x(\phi_x^{-1}(B)) \neq Q_x(\phi_x^{-1}(B)) = Q_z(B).
\]
Thus, if \(P_x \neq Q_x\), then \(P_z \neq Q_z\).

Now, consider the kernel mean embeddings under the composed kernel \(k_{\phi}\):
\[
\mu_{P_x}(\cdot) = \mathbb{E}_{x \sim P_x}[k_{\phi}(\cdot, x)] = \mathbb{E}_{x \sim P_x}[k(\phi_x(\cdot), \phi_x(x))].
\]

Since \(x \sim P_x\) implies \(\phi_x(x) \sim P_z\), we rewrite this embedding as:
\[
\mu_{P_x}(\cdot) = \mathbb{E}_{z \sim P_z}[k(\phi_x(\cdot), z)].
\]

This is precisely the kernel mean embedding of \(P_z\) in the RKHS associated with \(k\), evaluated at \(\phi_x(\cdot)\). Thus, we have:
\[
\mu_{P_x}(\cdot) = \mu_{P_z}(\phi_x(\cdot)),
\]
where \(\mu_{P_z}\) is the kernel mean embedding of \(P_z\) in the RKHS \(\mathcal{H}_z\) on \(\mathcal{Z}\).

As \(k\) is characteristic on \(\mathcal{Z}\), it follows that if \(P_x \neq Q_x\):
\[
\mu_{P_x}(\cdot) = \mu_{P_z}(\phi_x(\cdot)) \neq \mu_{Q_z}(\phi_x(\cdot)) = \mu_{Q_x}(\cdot).
\]

Suppose instead that \(\mu_{P_x} = \mu_{Q_x}\). Then, we have:
\[
\mu_{P_z}(\phi_x(\cdot)) = \mu_{Q_z}(\phi_x(\cdot)).
\]

This means that for all \(x \in \mathcal{X}\), \(\langle \mu_{P_z}, k(\cdot, \phi_x(x)) \rangle_{\mathcal{H}_z} = \langle \mu_{Q_z}, k(\cdot, \phi_x(x)) \rangle_{\mathcal{H}_z}\). Because \(k\) is characteristic, the span of the kernel functions \(\{k(\cdot, z) : z \in \mathcal{Z}\}\) is dense in \(\mathcal{H}_z\). Therefore, since the inner products of \(\mu_{P_z}\) and \(\mu_{Q_z}\) agree on a dense subset of \(\mathcal{H}_z\), they must be equal as functions: \(\mu_{P_z} = \mu_{Q_z}\). Because \(k\) is characteristic, this implies \(P_z = Q_z\). Finally, since \(\phi_x\) is injective, equality of pushforward measures \(P_z = Q_z\) implies equality of the original measures \(P_x = Q_x\).

Therefore, the composed kernel \(k_{\phi}\) is characteristic on \(\mathcal{X}\).
\end{proof}

\subsection{Hypothesis Testing}\label{hypothesis-testing}

To test the null hypothesis \(H_0: P = Q\), we use the empirical
\(\widehat{\text{MMD}}^2_u\) as our test statistic. To determine
statistical significance, we employ the permutation-based procedure
detailed in Algorithm \ref{alg:mmd_permutation_based}.

\begin{algorithm}[htbp]
\caption{Permutation-Based Hypothesis Test for MMD}
\label{alg:mmd_permutation_based}
\begin{algorithmic}[1]
\REQUIRE Samples \(X = \{x_1, \dots, x_m\}\) from distribution \(P\), samples \(Y = \{y_1, \dots, y_n\}\) from distribution \(Q\), kernel function \(k(\cdot,\cdot)\), number of permutation iterations \(P\) (e.g., \(P=1000\)), significance level \(\alpha\) (e.g., \(\alpha=0.01\)).
\STATE Compute the observed statistic:
    \[
    \Delta_{\text{obs}} \gets \widehat{\text{MMD}}^2_u(X, Y),
    \]
    where:
    \[
    \widehat{\text{MMD}}^2_u(X, Y) = \frac{1}{m(m-1)} \sum_{i=1}^{m}\sum_{\substack{j=1 \\ j\neq i}}^{m} k(x_i,x_j) + \frac{1}{n(n-1)} \sum_{i=1}^{n}\sum_{\substack{j=1 \\ j\neq i}}^{n} k(y_i,y_j) - \frac{2}{mn}\sum_{i=1}^{m}\sum_{j=1}^{n} k(x_i,y_j).
    \]
\STATE Pool samples into a single dataset of size \(m+n\):
    \[
    Z \gets X \cup Y.
    \]
\FOR{\(p = 1\) \TO \(P\)}
    \STATE Randomly permute the pooled sample \(Z\). Let the permuted sample be \(Z^*\).
    \STATE Partition \(Z^*\) into two sets: \(X^*_p\) containing the first \(m\) elements, and \(Y^*_p\) containing the remaining \(n\) elements.
    \STATE Compute the statistic on the permuted partition:
        \[
        \Delta^*_p \gets \widehat{\text{MMD}}^2_u(X^*_p, Y^*_p).
        \]
\ENDFOR
\STATE Calculate the p-value:
    \[
    p \gets \frac{1}{P}\sum_{p=1}^{P} \mathbf{1}\{\Delta^*_p \ge \Delta_{\text{obs}}\},
    \]
    where the indicator function is defined as:
    \[
    \mathbf{1}\{A\} = \begin{cases}
    1 & \text{if } A \text{ is true}\\[6pt]
    0 & \text{otherwise}
    \end{cases}.
    \]
\STATE Construct the \((100\times(1-\alpha))\%\) confidence interval from the permutation-based distribution:
    \[
    \left[Q_{\alpha/2}\left(\{\Delta^*_p\}_{p=1}^P\right), Q_{1-\alpha/2}\left(\{\Delta^*_p\}_{p=1}^P\right)\right],
    \]
    where \(Q_{\gamma}(\cdot)\) denotes the \(\gamma\)-quantile of the permutation-based statistics.
\IF{\(\Delta_{\text{obs}}\) falls outside the computed confidence interval (or equivalently, if \(p < \alpha\))}
    \STATE Reject \(H_0: P=Q\).
\ELSE
    \STATE Do not reject \(H_0\).
\ENDIF
\end{algorithmic}
\end{algorithm}

The algorithm provides a step-by-step procedure for resampling,
computing the MMD statistic under the null hypothesis, and calculating
the p-value and confidence interval. The confidence interval represents
the range of plausible values for the MMD statistic under the null
hypothesis. As described in the algorithm, if the observed statistic
\(\widehat{\text{MMD}}^2_u(X, Y)\) falls outside the confidence
interval---or equivalently, if the p-value is less than the significance
level \(\alpha\)---we reject the null hypothesis and conclude that the
distributions \(P\) and \(Q\) are statistically significantly different.

\section{Validation: MNIST Handwritten
Digits}\label{validation-mnist-handwritten-digits}

Before applying our MMD-based methodology to the legally salient domain
of AI-generated art, we first validate its statistical properties and
practical utility in a controlled setting with known ground truth. For
this purpose, we use the MNIST dataset of handwritten digits
(\citeproc{ref-lecun1998gradient}{LeCun et al. 1998}), a widely
recognized benchmark in machine learning. MNIST comprises 70,000
grayscale images (28×28 pixels) of handwritten digits from 0 to 9, split
into a training set of 60,000 images and a test set of 10,000 images
(containing approximately 1000 examples per digit class). Each image
represents a single digit, providing clear ground truth for our
validation: we know \emph{a priori} that the distributions of different
digits should be distinct.

To represent the images in a vector space suitable for MMD calculation,
we employ a convolutional neural network (CNN) embedding. Specifically,
we use the classic LeNet-5 architecture
(\citeproc{ref-lecun1998gradient}{LeCun et al. 1998}), a CNN designed
explicitly for handwritten digit recognition. LeNet-5 consists of two
convolutional layers with average pooling, followed by three fully
connected (dense) layers. The architecture details are summarized in
Table \ref{tab:lenet5}.

\begin{table}[htbp]
\centering
\begin{tabular}{lll}
\toprule
Layer Type                          & Output Shape & Parameters \\
\midrule
Input                               & 28×28×1      & 0          \\
Conv2D (6 filters, 5×5 kernel, ReLU)& 24×24×6      & 156        \\
AvgPool2D (2×2)                     & 12×12×6      & 0          \\
Conv2D (16 filters, 5×5 kernel, ReLU)& 8×8×16      & 2,416      \\
AvgPool2D (2×2)                     & 4×4×16       & 0          \\
Flatten                             & 256          & 0          \\
Dense (120 units, ReLU)             & 120          & 30,840     \\
Dense (84 units, ReLU)              & 84           & 10,164     \\
Dense (10 units, Softmax)           & 10           & 850        \\
\bottomrule
\end{tabular}
\caption{LeNet-5 Architecture Details}
\label{tab:lenet5}
\end{table}

We trained LeNet-5 on the MNIST training set using the Adam optimizer
(learning rate = 0.001), categorical cross-entropy loss, and a batch
size of 64. Training employed early stopping (patience = 10 epochs) and
model checkpointing (saving the best model based on validation loss).
The final trained model achieved excellent performance, with a test loss
of 0.0194 and test accuracy of 99.35\%, confirming its ability to
capture distinguishing visual features of each digit.

For our embedding, we extract the output of the second-to-last dense
layer (84 units) after passing each image through the trained LeNet-5
model. This provides an 84-dimensional vector representation for each
image, effectively mapping the high-dimensional image data into a
lower-dimensional space suitable for MMD analysis.

\subsection{MMD Analysis Procedure and
Setup}\label{mmd-analysis-procedure-and-setup}

Our validation procedure comprises the following steps:

\begin{enumerate}
\def\labelenumi{\arabic{enumi}.}
\item
  \textbf{Embedding Extraction:} We process all images from the MNIST
  \emph{test} set through our trained LeNet-5 model and extract the
  resulting 84-dimensional embeddings from the second-to-last dense
  layer. These embeddings represent each digit image as a numerical
  vector that captures the salient visual features identified by the
  neural network during training.
\item
  \textbf{Sample Generation:} For each digit pair (e.g., digit 0
  vs.~digit 1, digit 0 vs.~digit 2, etc.), we compile two separate sets
  of embeddings---one for each digit class. We then randomly sample
  (without replacement) specific quantities from these embedding sets
  for our analysis. To ensure balanced comparisons and maintain
  computational efficiency in the heatmap analysis (which involves all
  100 pairwise comparisons), we cap the sample size for each
  distribution at 400 embeddings, a substantial subset given the
  approximately 1000 available test samples per digit. As a negative
  control, we also compare samples drawn from the same digit class
  (e.g., digit 3 vs.~digit 3), where we expect the MMD statistic to be
  near zero and the null hypothesis not to be rejected. This provides a
  baseline for evaluating the method's false-positive rate.
\item
  \textbf{MMD Calculation and Hypothesis Testing:} For each digit pair
  and sample size, we compute the unbiased MMD statistic using a
  Gaussian radial basis function (RBF) kernel. We select the Gaussian
  RBF kernel for its characteristic property and ability to capture
  complex, nonlinear relationships between data points
  (\citeproc{ref-gretton2012kernel}{Gretton et al. 2012}). For the
  bandwidth parameter (\(\sigma\)) of the Gaussian kernel, we implement
  the median heuristic---setting \(\sigma\) to the median of all
  pairwise Euclidean distances in the combined sample. This
  data-adaptive approach provides a bandwidth that is both robust to
  outliers and appropriately scaled to the data's dimensionality. After
  calculating the MMD statistic, we perform the permutation-based
  hypothesis test described in Algorithm 1, using \(P=1000\) permutation
  iterations and a significance level of \(\alpha=0.01\). This test
  evaluates the null hypothesis that the two digit distributions are
  identical.
\item
  \textbf{Sample Size Variation and Rejection Rate Estimation:} To
  specifically stress-test the method's sensitivity and data efficiency,
  we repeat steps 2 and 3 across multiple \emph{very small} sample
  sizes---specifically 4, 5, 6, 7, 8, 9, 10, 12, 16, and 24. For each
  digit pair and each sample size in this range, we perform 100
  independent trials, each involving fresh random sampling and a full
  permutation test. Averaging the outcomes (reject/fail-to-reject
  \(H_0\)) across these 100 trials provides a robust estimate of the
  rejection rate (statistical power) for that specific scenario. We
  focus on a representative set of digit pairs that vary in visual
  similarity, including (0 vs.~1), (1 vs.~7), (2 vs.~8), (3 vs.~5), and
  (4 vs.~9), to evaluate the method's performance across both easy and
  challenging comparisons. This systematic exploration is essential for
  understanding the minimum data requirements for reliable distribution
  discrimination, particularly important for real-world applications
  where large datasets may be unavailable.
\end{enumerate}

By methodically varying both sample sizes and digit pairs, and employing
repeated trials for rejection rate estimation, we comprehensively
evaluate the sensitivity and reliability of our MMD-based approach under
different conditions. This thorough validation protocol ensures that our
method can effectively detect meaningful distributional differences,
even with limited available data---a critical consideration for
practical applications in novelty and distinctiveness assessment.
Anticipating the results, this protocol allows us to rigorously evaluate
the method's performance, expecting it to demonstrate high sensitivity
even with minimal data.

\subsection{Results: MNIST Validation
Study}\label{results-mnist-validation-study}

Figure \ref{fig:mnist_rejection_rate} illustrates the sensitivity of our
approach, displaying the estimated rejection rate of the null hypothesis
(\(H_0: P=Q\)) at a significance level of \(\alpha=0.01\) as the sample
size per distribution increases. The results demonstrate exceptional
data efficiency. For all digit pairs tested---including visually
distinct examples (e.g., 0 vs.~1, 1 vs.~7) and those exhibiting greater
visual similarity (e.g., 3 vs.~5, 4 vs.~9)---the rejection rate rapidly
surpasses the 95\% threshold at just n=6 samples per distribution, and
notably achieves this threshold with as few as n=5 samples for the digit
pair (2 vs.~8). This underscores the method's capability to capture
subtle yet statistically significant distributional differences, a
critically valuable feature in demanding contexts such as IP analyses,
where comprehensive data resources are frequently unavailable. As sample
size increases from \(n=4\) to \(n=24\), rejection rates uniformly
approach 100\% for all distinct digit pairs tested, confirming the
method's robust statistical convergence and reliability.

\begin{figure}[htbp]
\centering
\includegraphics[width=\linewidth]{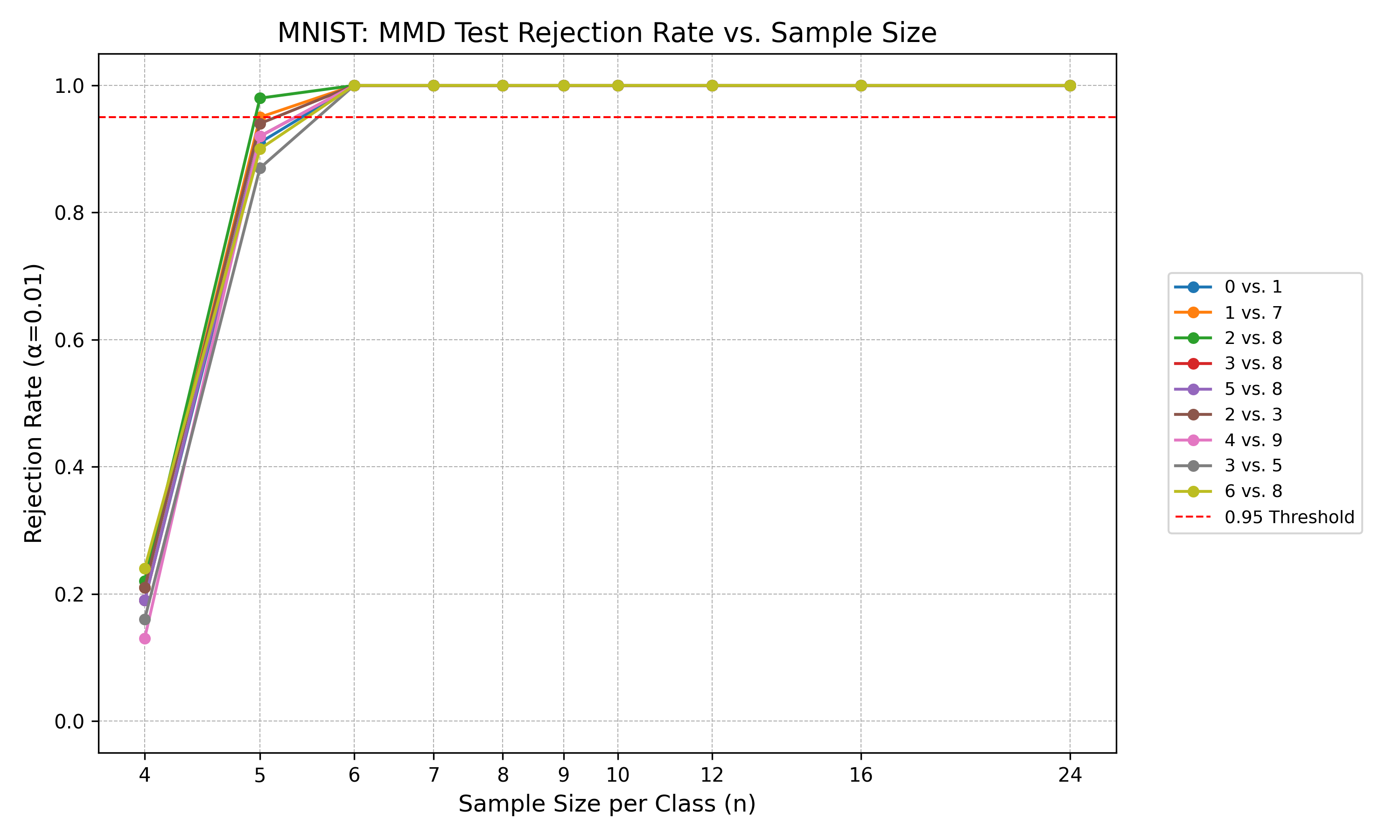}
\caption{Rejection Rate vs. Sample Size for Selected MNIST Digit Pairs}
\label{fig:mnist_rejection_rate}
\begin{minipage}{0.95\linewidth}
\footnotesize
Note: Each line represents the proportion of null hypothesis rejections (\(H_0: P=Q\)) at \(\alpha=0.01\), estimated by averaging results over 100 independent random sampling trials for each sample size and digit pair. The dashed line at 0.95 highlights rapid achievement of high statistical power with very small sample sizes (n=6 for all pairs shown).
\end{minipage}
\end{figure}

Figure \ref{fig:mnist_mmd_heatmap} complements this by depicting MMD
statistics across all digit comparisons at a sample size of \(n=400\).
Diagonal comparisons (negative controls, comparing samples of the same
digit) yield MMD statistics reliably close to zero (-0.0005 to 0.0025)
with uniformly non-significant results (p-values range from 0.0340 to
0.7010 at \(\alpha=0.01\)), confirming excellent control of
false-positive rates. In contrast, all off-diagonal comparisons (90 out
of 90 distinct digit pairs) differ statistically significantly
(\(p < 0.0001\)). Quantitatively, these significant differences range
from an MMD of 0.6558, observed between the visually similar digits 3
and 5, to a maximum MMD of 0.9703 between the highly distinct digits 0
and 3. This alignment between the quantitative MMD measure and intuitive
visual dissimilarity further validates the method's comprehensive
capability to numerically capture distributional differences, suggesting
clear applicability for legal and policy-relevant evaluations of
originality, distinctiveness, and novelty.

\begin{figure}[htbp]
\centering
\includegraphics[width=\linewidth]{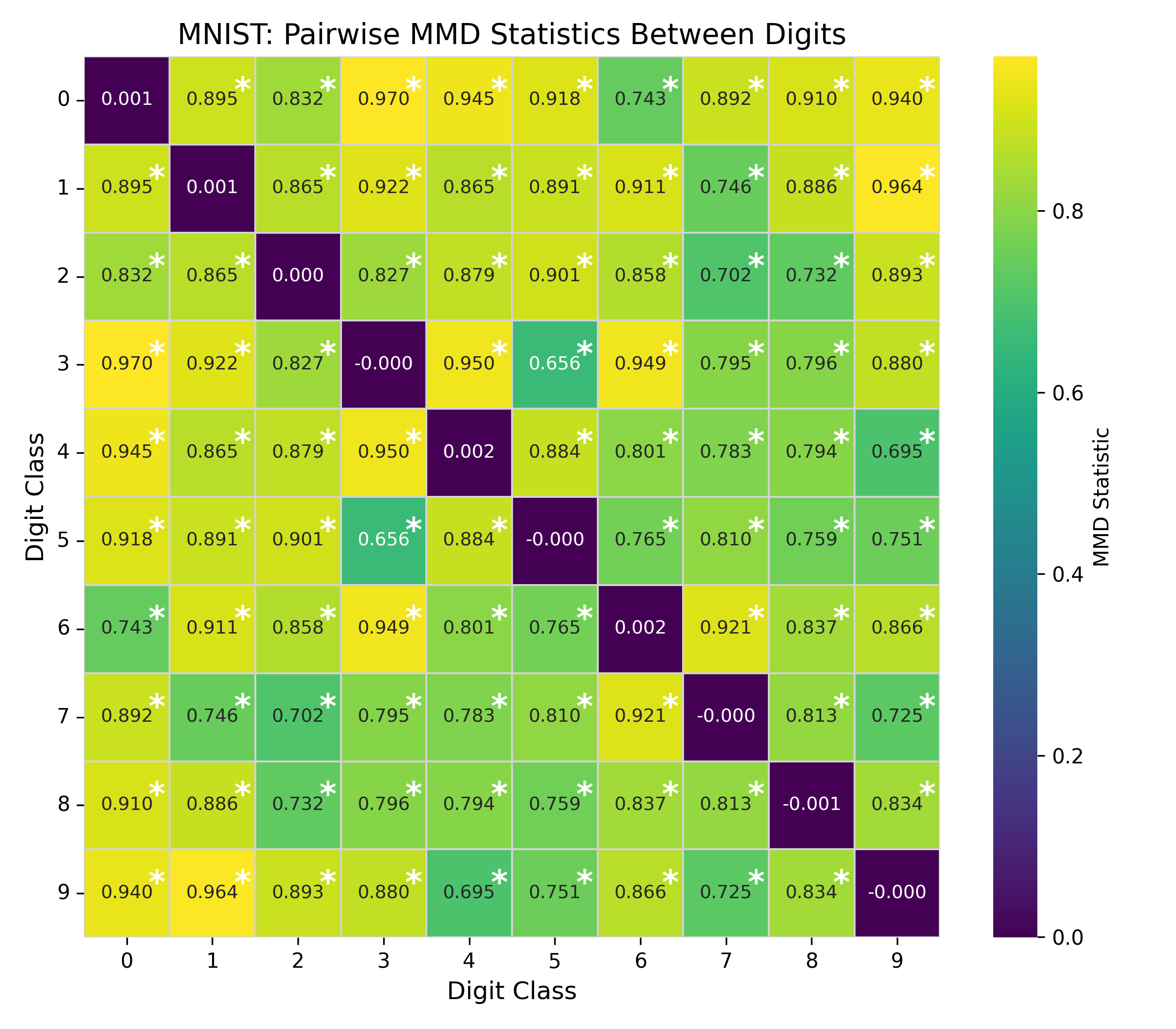}
\caption{Heatmap of MMD Statistics for All MNIST Digit Pairs (Sample Size \(n=400\)).}
\label{fig:mnist_mmd_heatmap}
\begin{minipage}{0.95\linewidth}
\footnotesize
Note: Diagonal cells (negative controls) show near-zero, non-significant MMD values. All off-diagonal cells show statistically significant differences (p < 0.0001, marked with *), with MMD magnitudes reflecting the degree of distributional dissimilarity.
\end{minipage}
\end{figure}

These results clearly demonstrate our methodology's exceptional ability
to reliably and efficiently distinguish between different digit
distributions, achieving statistically significant differentiation
(\(p < 0.01\)) with as few as 5 to 6 samples per digit class. This level
of data efficiency is particularly valuable in legal and policy
contexts, where comprehensive datasets may be unavailable or infeasible
to collect---such as when evaluating the novelty of a small set of
AI-generated works or comparing a new trademark to a limited set of
existing marks. Thus, our method's combination of sensitivity,
statistical rigor, and quantitative interpretability makes it a strong
candidate for real-world applications.

Having established this statistical foundation, we now apply the
methodology to the more complex and legally salient domain of
AI-generated art. We note that while this validation used LeNet-5
embeddings and a Gaussian RBF kernel, the framework's flexibility allows
for alternative choices, the impact of which may warrant future
investigation.

\section{AI-Generated Art -- Distinguishing Human and Machine
Creativity}\label{ai-generated-art-distinguishing-human-and-machine-creativity}

Having established the validity and sensitivity of our MMD-based
methodology in the controlled environment of the MNIST dataset, we now
turn to a more complex and nuanced real-world application:
distinguishing between human-created and AI-generated art. Whereas MNIST
demonstrated the method's effectiveness in a domain with clear classes,
art is inherently subjective, stylistically diverse, and lacks simple
ground truth, presenting a far greater challenge for automated analysis.
This application, therefore, directly addresses the core research
question: Can our MMD-based approach statistically distinguish
AI-generated art distributions from human-created art distributions,
even when visual differences become increasingly subtle? Answering this
has direct implications for legal questions surrounding the originality
and distinctiveness of AI outputs.

\subsection{The AI-ArtBench Dataset and
Categories}\label{the-ai-artbench-dataset-and-categories}

To investigate this question, we utilize the AI-ArtBench dataset
(\citeproc{ref-silva2024artbrain}{Silva et al. 2024}), a comprehensive
collection designed specifically for studying AI-generated imagery. It
comprises 185,015 artistic images spanning 10 distinct art styles (e.g.,
Impressionism, Surrealism, Art Nouveau). Crucially for our study, this
dataset includes both human-created artworks (60,000 images derived from
the rigorously curated ArtBench-10 dataset
(\citeproc{ref-liao2022artbench}{Liao et al. 2022})) and AI-generated
images (125,015 images) produced using text prompts based on the human
artworks. The AI images were generated by two different, prominent
diffusion models:

\begin{itemize}
\tightlist
\item
  \textbf{Standard Diffusion (SD):} A widely used diffusion model
  operating in pixel space.
\item
  \textbf{Latent Diffusion (LD):} A more recent diffusion model
  operating in a lower-dimensional latent space, often associated with
  higher perceived quality and diversity.
\end{itemize}

For our analysis, we categorize the images into three distinct groups:
\texttt{Human} (original human artworks), \texttt{AI\ (SD)} (images
generated by Standard Diffusion), and \texttt{AI\ (LD)} (images
generated by Latent Diffusion). The inclusion of two distinct AI
generation methods is important: it allows us not only to compare
AI-generated art to human art but also to test whether our MMD
methodology is sensitive enough to detect potential distributional
differences \emph{between} different AI generation techniques
themselves.

The AI-ArtBench dataset is particularly valuable because recent research
using it has shown that humans struggle to reliably distinguish between
human and AI-generated art, achieving only approximately 58\% accuracy
in an ``Artistic Turing Test'' (\citeproc{ref-silva2024artbrain}{Silva
et al. 2024}). This highlights the increasingly blurred line between
human and machine creativity in the visual arts, at least to the human
eye, and underscores the need for robust, quantitative methods capable
of detecting potential underlying distributional differences.

\subsection{Embedding with CLIP for Semantic
Representation}\label{embedding-with-clip-for-semantic-representation}

Unlike the MNIST dataset, where a specialized CNN (LeNet-5) trained
specifically for digit recognition was appropriate, analyzing art
requires capturing more complex visual styles, themes, and semantic
content. Simple pixel-level comparisons or features learned for narrow
tasks are insufficient. Moreover, in realistic legal and policy
contexts, labeled datasets specifically tailored to distinguish
AI-generated art from human-created art are typically unavailable or
prohibitively expensive to create. Consequently, training a specialized
embedding model from scratch for each new comparison would be
impractical.

To address both the need for semantic richness and this practical
constraint of data availability, we employ a general-purpose embedding
method using the CLIP (Contrastive Language-Image Pre-training) model
(\citeproc{ref-radford2021learning}{Radford et al. 2021}). CLIP is a
powerful neural network architecture pre-trained on a massive dataset of
image-text pairs, learning representations that align visual and textual
concepts. Its pre-trained nature allows it to be applied
``off-the-shelf'' without requiring bespoke, task-specific training
data. Furthermore, CLIP embeddings capture both visual features and
higher-level semantic information, positioning images with similar
styles, subjects, and artistic concepts closer together in the embedding
space. This combination of practical applicability and semantic depth is
crucial for capturing the nuances of artistic expression needed for our
distributional analysis.

Specifically, we utilize the \texttt{ViT-H-14-quickgelu} variant of
CLIP, pre-trained on the large-scale \texttt{dfn5b} dataset, accessed
via the \texttt{open\_clip} library. This model offers a strong balance
between representational power and computational feasibility. We process
each selected image from the AI-ArtBench dataset through the pre-trained
CLIP image encoder. The output for each image is its corresponding
embedding vector, which we normalize to unit length. These normalized
embeddings, which are 1024-dimensional vectors, serve as the input data
points for our subsequent MMD analysis. The reliance on such a
pre-trained, general embedding makes our MMD framework readily
applicable for assessing distinctiveness across diverse image sets
without requiring domain-specific fine-tuning---a key advantage for
timely legal and policy evaluations where the ability to quickly and
reliably compare new image sources is paramount.

\subsection{MMD Analysis Procedure and
Setup}\label{mmd-analysis-procedure-and-setup-1}

Following the successful validation on MNIST, we apply the same core MMD
methodology to the AI-ArtBench embeddings, adapting the procedure for
the three categories (\texttt{Human}, \texttt{AI\ (SD)},
\texttt{AI\ (LD)}). The key steps are as follows:

\begin{enumerate}
\def\labelenumi{\arabic{enumi}.}
\item
  \textbf{Data Loading and Sampling:} We load images from the
  \texttt{test} split of the AI-ArtBench dataset, specifically targeting
  the directories corresponding to our three categories (\texttt{Human},
  \texttt{AI\ (SD)}, \texttt{AI\ (LD)}). To ensure balanced comparisons
  between categories and manage computational load for embedding
  extraction and MMD calculations, we randomly sample (without
  replacement) a maximum of 3000 images per category, resulting in a
  total dataset of up to 9000 images (3000 per category). This provides
  a substantial yet manageable dataset for analysis.
\item
  \textbf{Embedding Extraction:} As described previously, we pass each
  of the sampled images through the pre-trained CLIP model
  (\texttt{ViT-H-14-quickgelu}, \texttt{dfn5b} pre-training) to obtain
  its normalized 1024-dimensional embedding vector. This results in
  three distinct sets of embedding vectors, one for each category.
\item
  \textbf{Pairwise MMD Calculation and Hypothesis Testing:} We compute
  the unbiased MMD statistic (using the identical Gaussian RBF kernel
  with the median heuristic for bandwidth selection as in the MNIST
  study) between all unique pairs of categories: \texttt{Human}
  vs.~\texttt{AI\ (SD)}, \texttt{Human} vs.~\texttt{AI\ (LD)}, and
  \texttt{AI\ (SD)} vs.~\texttt{AI\ (LD)}. We also compute the MMD for
  each category against itself (\texttt{Human} vs.~\texttt{Human}, etc.)
  by splitting the category's samples into two random halves; these
  serve as crucial negative controls. In these negative-control
  comparisons, we expect near-zero MMD values and non-significant
  results, confirming that the method does not falsely detect
  differences when comparing identical distributions. For the main
  pairwise comparisons used in the heatmap (Figure
  \ref{fig:art_mmd_heatmap}), we use a sample size capped at n=400 per
  category (drawn from the available 3000) for computational efficiency,
  consistent with the MNIST heatmap approach. For each comparison, we
  perform the permutation-based hypothesis test (Algorithm 1) with
  P=2500 permutation iterations and a significance level of
  \(\alpha=0.01\) to determine if the observed MMD is statistically
  significant, evaluating the null hypothesis \(H_0: P=Q\).
\item
  \textbf{Sample Size Variation and Rejection Rate Estimation:} To
  assess the method's sensitivity and data efficiency in this more
  complex domain, we repeat the MMD calculation and permutation test for
  the three \emph{off-diagonal} pairwise comparisons (\texttt{Human}
  vs.~\texttt{AI\ (SD)}, etc.) across a range of small sample sizes: n =
  4, 5, 6, 7, 8, 9, 10, 12, 16, and 24. Similar to the MNIST procedure,
  we perform 100 independent trials for each pair at each sample size,
  averaging the test outcomes to estimate the rejection rate
  (statistical power) as plotted in Figure \ref{fig:art_rejection_rate}.
  The maximum sample size per category for these trials is capped at the
  overall heatmap cap (n=400) if the specified sample size \texttt{n}
  exceeds it.
\end{enumerate}

This structured procedure allows us to rigorously test whether the
distributions of human and AI-generated art, as represented by CLIP
embeddings, are statistically distinguishable, and how much data is
required to reliably detect such differences.

\subsection{Results: AI-ArtBench Study}\label{results-ai-artbench-study}

Applying the described MMD analysis procedure yields clear quantitative
evidence regarding the distributional differences between human-created
and AI-generated art within the AI-ArtBench dataset, as represented by
CLIP embeddings.

Figure \ref{fig:art_mmd_heatmap} presents the 3×3 MMD heatmap,
summarizing the pairwise comparisons between the \texttt{Human},
\texttt{AI\ (SD)}, and \texttt{AI\ (LD)} categories using a sample size
of \(n=400\) per category. As expected, the diagonal elements (negative
controls) show MMD statistics very close to zero (ranging from -0.0001
to 0.0005) and are uniformly non-significant (p-values \textgreater{}
0.14), confirming the test's reliability under the null hypothesis.

\begin{figure}[htbp]
\centering
\includegraphics[width=\linewidth]{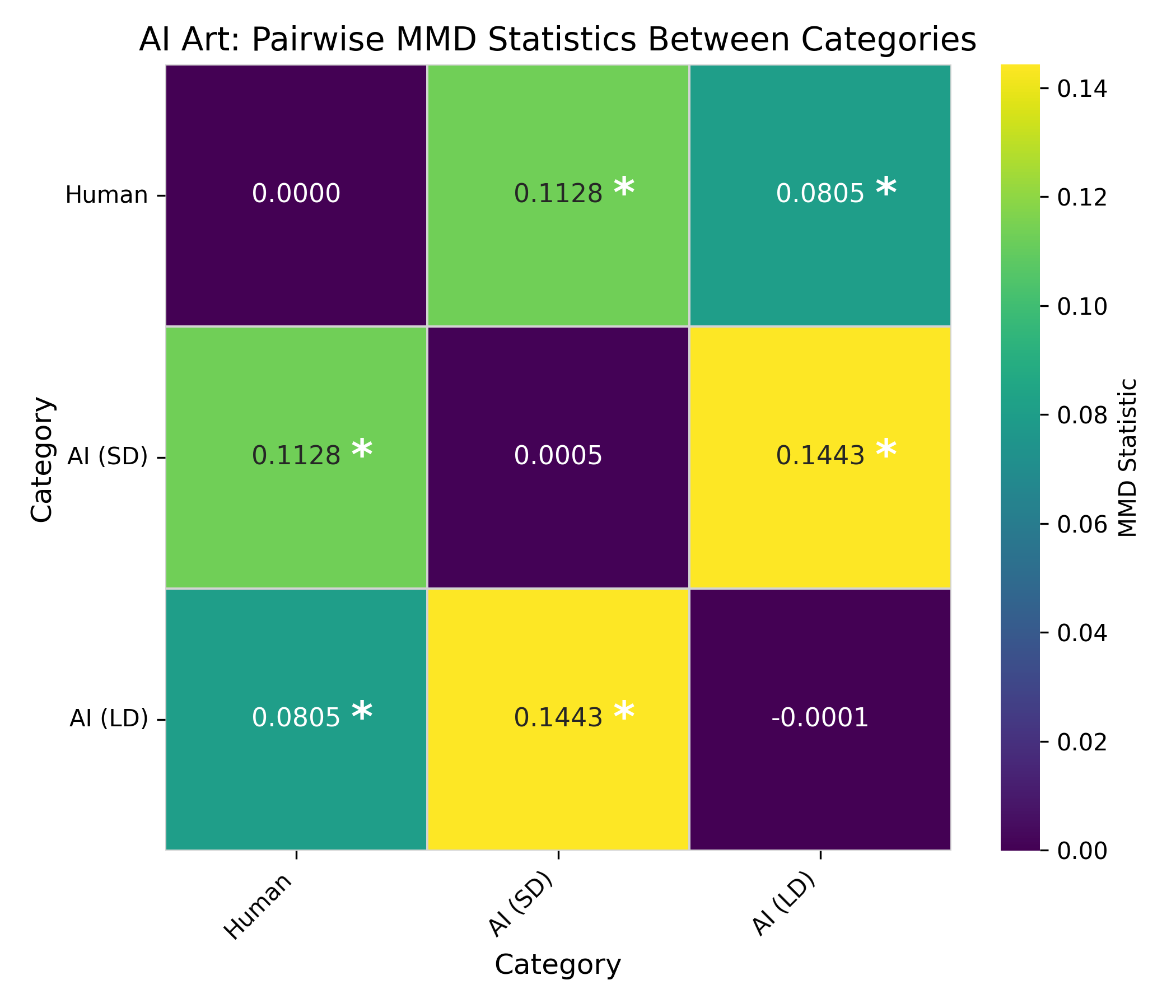}
\caption{Heatmap of MMD Statistics for AI-ArtBench Categories (Sample Size \(n=400\)).}
\label{fig:art_mmd_heatmap}
\begin{minipage}{0.95\linewidth}
\footnotesize
Note: `Human' indicates original human artworks, `AI (SD)' indicates images generated by Standard Diffusion, and `AI (LD)' indicates images generated by Latent Diffusion. Diagonal cells (negative controls) show near-zero, non-significant MMD values (p-values range from 0.1448 to 0.5284). All off-diagonal cells show statistically significant differences (p < 0.0001, marked with *), with MMD magnitudes reflecting the degree of distributional dissimilarity based on 1024-dim CLIP embeddings.
\end{minipage}
\end{figure}

Specifically, both comparisons between human-created art and
AI-generated art yield statistically significant MMD values: the
comparison between \texttt{Human} and \texttt{AI\ (SD)} produces an MMD
of 0.1128 (p \textless{} 0.0001), and the comparison between
\texttt{Human} and \texttt{AI\ (LD)} yields an MMD of 0.0805 (p
\textless{} 0.0001). These results indicate that the distributions of
CLIP embeddings for both AI models are statistically distinguishable
from the distribution of human-created art. Notably, these distinctions
are evident for both SD and LD images---even though humans struggle to
visually distinguish them from human-made artwork, achieving only
approximately 58\% accuracy in the Artistic Turing Test
(\citeproc{ref-silva2024artbrain}{Silva et al. 2024}). This underscores
the sensitivity of our distributional approach, capable of detecting
subtle semantic differences that may elude direct human perception.

Interestingly, the magnitude of the differences between human-created
art and each AI-generated category is remarkably similar (0.0805
vs.~0.1128), suggesting that within the CLIP embedding space, both AI
generation methods diverge from the human art distribution to a
comparable extent. Furthermore, the comparison between the two AI models
(\texttt{AI\ (SD)} vs.~\texttt{AI\ (LD)}) also yields a statistically
significant difference (MMD = 0.1443, p \textless{} 0.0001). However,
this MMD value is larger than the human-AI differences. This suggests
the two AI generation processes, while both distinct from human art, are
more dissimilar from each other in the CLIP embedding space than either
is to the human-created art distribution. This finding highlights the
method's sensitivity in capturing nuanced differences even between
different generative processes.

Figure \ref{fig:art_rejection_rate} further explores the sensitivity and
data efficiency of the MMD test by showing the rejection rate
(\(H_0: P=Q\)) as a function of sample size for the three pairwise
comparisons. The rejection rate increases rapidly with sample size for
all three comparisons, quickly approaching 100\%. Remarkably, sample
sizes ranging from n=7 (for Human vs.~AI SD) to n=10 (for Human vs.~AI
LD) images per category are sufficient to reliably distinguish between
the pairs (\texttt{Human} vs.~\texttt{AI\ (SD)}, \texttt{Human}
vs.~\texttt{AI\ (LD)}, and \texttt{AI\ (SD)} vs.~\texttt{AI\ (LD)}) with
over 95\% confidence (\(p < 0.01\)). Although this convergence
(requiring 7-10 samples here) is slightly slower than observed in the
simpler MNIST domain (where only 5--6 samples were required), this
difference is expected given the greater complexity, subtlety, and
subjective variability inherent in artistic images. Nonetheless,
achieving reliable statistical discrimination with fewer than a dozen
samples per category remains exceptionally data-efficient, underscoring
the practical utility of our method in real-world scenarios where data
availability may be limited.

\begin{figure}[htbp]
\centering
\includegraphics[width=\linewidth]{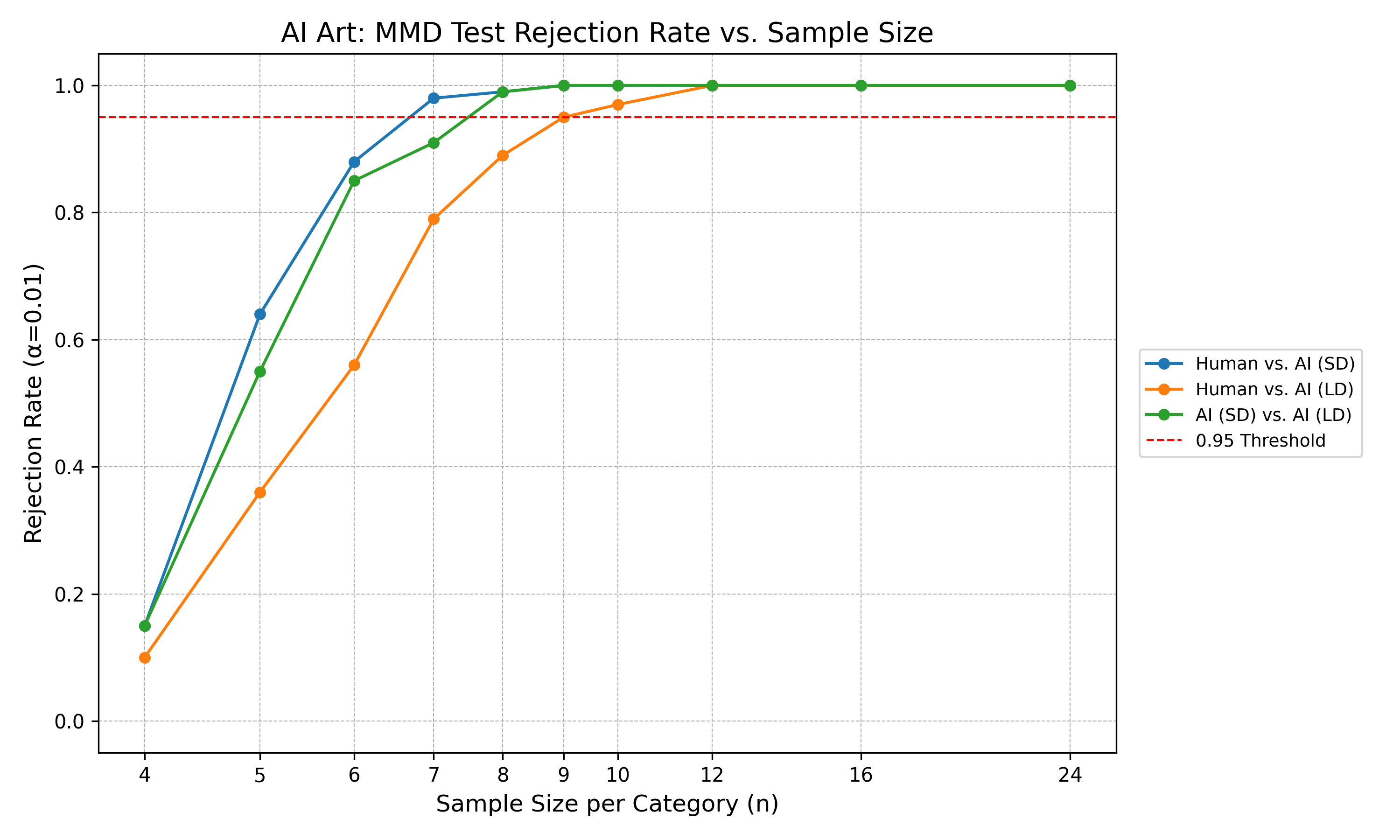}
\caption{Rejection Rate vs. Sample Size for AI-ArtBench Category Comparisons}
\label{fig:art_rejection_rate}
\begin{minipage}{0.95\linewidth}
\footnotesize
Note: `Human' indicates original human artworks, `AI (SD)' indicates images generated by Standard Diffusion, and `AI (LD)' indicates images generated by Latent Diffusion. Each line represents the proportion of null hypothesis rejections (\(H_0: P=Q\)) at \(\alpha=0.01\), estimated over 100 independent trials per point. The dashed line at 0.95 highlights rapid achievement of high statistical power; all pairs reach >95\% rejection rate with only n=7-10 samples per category.
\end{minipage}
\end{figure}

\subsection{Conclusions: Distinguishing Human and Machine
Creativity}\label{conclusions-distinguishing-human-and-machine-creativity}

These results provide strong quantitative evidence that, within the
semantic space captured by CLIP embeddings, AI-generated art (from both
SD and LD models) forms distributions that are statistically distinct
from the distribution of human-created art in the AI-ArtBench dataset.
Whereas the MNIST study demonstrated the method's effectiveness in a
simpler domain with clearly defined classes, the success on AI-ArtBench
underscores the method's robustness in a highly subjective and
stylistically diverse creative domain. This finding directly challenges
simplistic notions of AI art as merely a ``stochastic parrot'' perfectly
mimicking human creativity; while trained on human data, the resulting
output distributions exhibit measurable differences.

Furthermore, the ability to distinguish between the two AI models,
albeit with a smaller MMD, highlights the potential for this methodology
to track and characterize the outputs of evolving generative techniques.
The use of powerful semantic embeddings like CLIP is crucial for
capturing the relevant stylistic and content nuances of art. Combined
with the MMD framework, this provides a sensitive, data-efficient, and
statistically rigorous tool for analyzing the distributional properties
of AI-generated content. The fact that these statistically significant
differences are readily detectable with very small sample sizes
(\(n=7-10\)), even when human visual discrimination is poor
(\textasciitilde58\% accuracy), emphasizes the power of distributional
analysis and its potential relevance for legal and policy discussions
regarding AI novelty and originality. This application demonstrates the
practical utility of our methodology for complex, real-world problems,
offering a foundation for further research into the nature of AI
creativity and its implications across various domains, including
intellectual property and art authentication. We next discuss how these
findings inform broader policy questions regarding AI originality and
distinctiveness.

\section{General Discussion}\label{general-discussion}

This paper develops and validates a novel, distribution-based
methodology for quantifying the novelty, originality, and
distinctiveness of a set of content given prior art--a baseline set of
content. Our approach addresses a fundamental challenge in current legal
frameworks: while IP law relies heavily on concepts of novelty,
distinctiveness, and originality, traditional assessment methods based
on pairwise comparisons or simple aggregations like average similarity
are fundamentally inadequate for evaluating AI outputs against an
effectively unbounded body of prior art. Unlike methods focusing on
individual item similarity, our MMD-based framework captures differences
in the underlying \emph{distributions} of creative processes. By
combining kernel mean embeddings, maximum mean discrepancy, and
domain-specific machine learning embeddings, we provide courts and
policymakers with a principled alternative that aligns with established
legal principles while accommodating the unique challenges posed by
generative AI.

Our methodology offers three key advantages that make it particularly
valuable for legal applications: it requires no model-specific training,
making it adaptable to evolving AI technologies; it operates effectively
with limited samples, addressing the practical reality that
comprehensive datasets are often unavailable in legal contexts; and it
yields statistically rigorous measures of distributional difference that
can inform legal determinations of originality and distinctiveness.
These properties enable the method to serve as a quantitative tool for
courts and IP offices evaluating AI-generated works, providing an
empirical foundation for legal reasoning that traditionally relies on
more subjective assessments. This approach, in line with Prakken
(\citeproc{ref-prakken2010abstract}{2010})'s advocacy for formal,
argument-based reasoning in legal analysis, demonstrates how
computational methods can systematize legal analysis by offering a more
objective framework for decision-making.

Through rigorous empirical validation, we provide compelling evidence
that AI-generated outputs can be statistically distinct from prior art,
even when the AI is explicitly prompted to generate content that
maximizes commercial viability and thus should encourage similarity to
successful precedents. This finding directly challenges the ``stochastic
parrot'' critique that has significantly influenced legal discourse
surrounding AI creativity and has been cited in ongoing copyright
litigation. Our results demonstrate that modern AI systems do not merely
mimic training data but produce semantically distinct outputs that may
warrant legal recognition.

Through rigorous empirical validation, we provide compelling evidence
that AI-generated outputs can be statistically distinct from prior art,
even when the AI is explicitly prompted to generate content that
maximizes commercial viability and thus should encourage similarity to
successful precedents. This finding directly challenges the ``stochastic
parrot'' critique that has significantly influenced legal discourse
surrounding AI creativity and has been cited in ongoing copyright
litigation. Our results demonstrate that modern AI systems do not merely
mimic training data but produce semantically distinct
outputs.\footnote{This statistical distinctiveness must be distinguished
  from \emph{legal originality} or \emph{transformativeness}. While our
  method provides objective evidence against claims of mere mimicry,
  legal determinations hinge on additional factors, including specific
  authorship requirements, the nature of the creative contribution, the
  idea/expression dichotomy, and fair use considerations. MMD offers
  valuable quantitative evidence \emph{for} these legal assessments,
  rather than replacing them.}

The implications of these findings are profound for current intellectual
property regimes. Existing frameworks, predicated on human authorship
and creativity, struggle to accommodate AI-generated works that exhibit
measurable novelty without direct human creative intent. Our
category-specific analysis further underscores this point, revealing
that AI's creative tendencies vary systematically across domains, with
stronger alignment to certain creative fields (e.g., fiction) than
others (e.g., culinary arts)---a finding with direct relevance to
domain-specific IP protections. Indeed, as AI-generated content
increasingly challenges traditional legal concepts of authorship and
originality, scholars have argued that integrating artificial
intelligence into legal reasoning itself may necessitate rethinking
fundamental legal doctrines and frameworks
(\citeproc{ref-verheij2020artificial}{Verheij 2020}). For trademark law,
our analysis of brand name distinctiveness addresses emerging concerns
about AI-generated marks. For copyright law, our methodology provides a
quantitative approach to assessing the traditionally qualitative concept
of originality. For patent law, it offers a potential tool for
evaluating non-obviousness in AI-generated inventions.\footnote{However,
  our findings that AI can produce statistically distinct outputs
  intersect with profound challenges to the existing non-obviousness
  standard itself. As scholars like Abbott
  (\citeproc{ref-abbott2019everything}{2019}) argue, if AI systems
  become standard tools for innovation, the benchmark ``person having
  ordinary skill in the art'' (PHOSITA) may need to be redefined to
  incorporate AI capabilities. The very definition of what is
  ``obvious'' relative to an AI-augmented PHOSITA may need
  re-evaluation, as AI improves and increasingly renders innovative
  activities ``obvious''. Our quantitative evidence of AI's capacity for
  generating distinct outputs lends empirical weight to the urgency of
  addressing how the non-obviousness doctrine should adapt to
  technologies that can systematically explore and generate solutions
  previously considered inventive.}

While this study provides strong evidence for AI's capacity for novelty,
we acknowledge limitations relevant to legal applications. The
effectiveness of our method depends on the quality of the chosen
embedding and kernel function, paralleling how legal determinations
depend on the frameworks used to evaluate creative works. Our analysis
also treated prior art as static; in reality, prior art is dynamic,
especially as AI-generated content increasingly enters the public
domain---a complexity that future legal frameworks must address.
Additionally, our method measures distinctiveness but does not directly
assess other legally relevant factors such as creative value or intent.
Translating MMD scores into discrete legal judgments will require
further consideration and the development of context-specific
guidelines.

Future research at this critical intersection should prioritize: (1)
establishing threshold MMD values that correspond to legal standards of
originality and distinctiveness across different IP domains; (2)
exploring how this methodology can be adapted to assess the novelty of
works created through human-AI collaboration, which present particularly
complex questions of authorship; (3) investigating how courts might
incorporate distributional evidence of novelty within existing legal
tests, addressing the interpretability challenges; and (4) developing
comprehensive legal frameworks that appropriately balance recognition of
AI's novel contributions with the broader social and economic goals of
intellectual property protection.

By providing both a robust methodological foundation and compelling
empirical evidence, this work contributes to a more nuanced
understanding of AI as a creative force within legal frameworks. As
courts and policymakers continue to grapple with rapid advancements in
generative AI capabilities, our approach offers a principled analytical
tool to help ground legal debates in quantitative evidence, ensuring
intellectual property law evolves in ways that accurately reflect
technological realities while preserving its fundamental purposes of
incentivizing innovation and creative expression.

\newpage

\section{Bibliography}\label{bibliography}

\singlespacing

\phantomsection\label{refs}

\newpage

\setcounter{section}{0}
\renewcommand{\thesection}{\Alph{section}}
\renewcommand{\thesubsection}{\thesection.\arabic{subsection}}
\renewcommand{\thesubsubsection}{\thesubsection.\arabic{subsubsection}}
\renewcommand{\theparagraph}{\thesubsubsection.\arabic{paragraph}}
\renewcommand{\thesubparagraph}{\theparagraph.\arabic{subparagraph}}
\renewcommand{\thetable}{A\arabic{table}}

\section{Web Appendix A: Python Code
Implementation}\label{web-appendix-a-python-code-implementation}

This appendix provides the complete Python implementation used to
operationalize and validate the Maximum Mean Discrepancy (MMD)-based
methodology developed in this paper. The code directly supports the
empirical analyses presented in Section 3 (MNIST validation study) and
Section 4 (AI-generated art study using the AI-ArtBench dataset). It is
structured to ensure full reproducibility of our findings and to serve
as a practical, general-purpose tool that researchers can adapt for
quantitative novelty and distinctiveness analysis in other domains
relevant to legal, scientific, or creative inquiries.

The implementation is organized into five distinct sections:

\subsection{Section 1: Shared MMD and Permutation Test
Functions}\label{section-1-shared-mmd-and-permutation-test-functions}

This foundational section defines the core statistical engine
underpinning the entire analysis. These functions are domain-agnostic,
implementing the MMD-based hypothesis testing framework detailed
theoretically in Section 2 of the main paper. They provide the reusable
tools for comparing distributions based on sample data.

\begin{itemize}
\tightlist
\item
  \textbf{Key Functions:}

  \begin{itemize}
  \tightlist
  \item
    \texttt{\_compute\_sigma\_median\_heuristic(x,\ y)}: A helper
    function that automatically determines an appropriate bandwidth
    parameter (\(\sigma\)) for the Gaussian Radial Basis Function (RBF)
    kernel. It uses the median heuristic, a standard data-driven
    approach that adapts the kernel's sensitivity to the scale of the
    input embeddings.
  \item
    \texttt{mmd\_squared\_unbiased(x,\ y,\ kernel,\ sigma)}: Calculates
    the unbiased estimate of the squared MMD statistic. This is the core
    measure quantifying the distance between the probability
    distributions from which sample sets \texttt{x} and \texttt{y} are
    drawn. The function supports both the flexible RBF kernel (default)
    and a simpler linear kernel.
  \item
    \texttt{permutation\_test(x,\ y,\ P,\ kernel,\ sigma,\ alpha,\ n\_jobs)}:
    Implements the non-parametric permutation test described in
    Algorithm 1. This function assesses the statistical significance of
    the observed MMD value (\texttt{delta\_obs}) by comparing it to a
    distribution of MMD values computed under the null hypothesis (H0:
    P=Q). The null distribution is generated by repeatedly shuffling the
    combined data (\texttt{x} and \texttt{y}) and recalculating MMD
    (\texttt{P} times). It returns the p-value and determines whether to
    reject H0 at the specified significance level \texttt{alpha}.
    Parallel processing (\texttt{n\_jobs}) is used to accelerate the
    computationally intensive permutation process.
  \end{itemize}
\end{itemize}

\subsection{Section 2: MNIST Validation Study
Functions}\label{section-2-mnist-validation-study-functions}

This section contains all code specifically designed for the MNIST
validation study (Section 3 of the main paper). The purpose here is to
demonstrate the MMD methodology's effectiveness and statistical
properties (like data efficiency and control of false positives) in a
controlled environment where the ground truth is known (i.e., images of
different handwritten digits \emph{should} come from distinct
distributions).

\begin{itemize}
\tightlist
\item
  \textbf{Key Functions:}

  \begin{itemize}
  \tightlist
  \item
    \texttt{mnist\_load\_and\_prepare\_data()}: Handles loading the
    standard MNIST dataset, performing necessary preprocessing
    (normalization, reshaping), and splitting it into training,
    validation, and test sets.
  \item
    \texttt{mnist\_build\_lenet5\_model()}: Defines the LeNet-5
    convolutional neural network (CNN) architecture, a classic benchmark
    model for digit recognition. The output of its penultimate layer
    (\texttt{embedding\_layer}) is used to generate numerical vector
    representations (embeddings) of the digit images.
  \item
    \texttt{mnist\_train\_model(...)}: Trains the LeNet-5 model on the
    MNIST training data. Includes standard practices like data
    augmentation (to improve robustness), early stopping (to prevent
    overfitting), and model checkpointing (to save the best performing
    model).
  \item
    \texttt{mnist\_evaluate\_model(...)}: Assesses the trained model's
    accuracy and loss on the unseen test set, confirming its ability to
    distinguish digits.
  \item
    \texttt{mnist\_extract\_embeddings(...)}: Uses the trained LeNet-5
    model to convert the MNIST test images into 84-dimensional embedding
    vectors, suitable for MMD analysis.
  \item
    \texttt{mnist\_compute\_rejection\_rates(...)}: Systematically
    evaluates the MMD test's statistical power. It runs the
    \texttt{permutation\_test} repeatedly
    (\texttt{N\_TRIALS\_REJ\_RATE}) for specified digit pairs across a
    range of small sample sizes (\texttt{SAMPLE\_SIZES}) and calculates
    the proportion of times the null hypothesis is correctly rejected.
    This demonstrates the method's sensitivity with limited data.
  \item
    \texttt{mnist\_compute\_mmd\_matrix(...)}: Computes the full 10×10
    matrix containing the MMD statistic and corresponding p-value for
    every pair of digit classes (0-9). This includes diagonal
    comparisons (e.g., `3' vs `3') as negative controls.
  \item
    \texttt{mnist\_plot\_rejection\_rates(...)} \&
    \texttt{mnist\_plot\_mmd\_heatmap(...)}: Generate the key
    visualizations presented in the paper: the plot showing how
    rejection rates increase with sample size, and the heatmap
    illustrating the MMD values between all digit pairs, annotated with
    significance markers.
  \item
    \texttt{mnist\_print\_summary\_statistics(...)}: Outputs a formatted
    text table summarizing the MMD results, clearly distinguishing
    negative controls from pairwise comparisons and indicating
    statistical significance.
  \end{itemize}
\end{itemize}

\subsection{Section 3: AI Art Study
Functions}\label{section-3-ai-art-study-functions}

This section applies the validated MMD methodology to the more complex
and legally relevant domain of AI-generated art, using the AI-ArtBench
dataset (Section 4 of the main paper). It compares human-created art
with AI-generated art produced by two different diffusion models
(Standard Diffusion - SD, Latent Diffusion - LD).

\begin{itemize}
\tightlist
\item
  \textbf{Key Functions:}

  \begin{itemize}
  \tightlist
  \item
    \texttt{art\_load\_dataset(...)}: Loads images from the AI-ArtBench
    dataset directory structure, correctly identifying and categorizing
    images into `Human', `AI (SD)', and `AI (LD)' groups based on folder
    names. It includes sampling logic to handle potentially large
    datasets.
  \item
    \texttt{art\_extract\_clip\_embeddings(...)}: Extracts
    high-dimensional (1024-dim) semantic embeddings for each artwork
    using a powerful, pre-trained CLIP model
    (\texttt{ViT-H-14-quickgelu}). CLIP is chosen here because its
    embeddings capture richer semantic and stylistic information
    necessary for comparing complex visual art, unlike the simpler
    features sufficient for MNIST digits. Embeddings are normalized.
  \item
    \texttt{art\_compute\_mmd\_matrix(...)}: Computes the 3×3 matrix of
    pairwise MMD statistics and p-values between the `Human', `AI (SD)',
    and `AI (LD)' categories using their CLIP embeddings.
  \item
    \texttt{art\_compute\_rejection\_rates(...)}: Similar to the MNIST
    study, this calculates the rejection rate of the MMD test for the
    three crucial pairwise comparisons (Human vs.~AI SD, Human vs.~AI
    LD, AI SD vs.~AI LD) across the specified range of small sample
    sizes, again demonstrating data efficiency in this harder task.
  \item
    \texttt{art\_plot\_mmd\_heatmap(...)} \&
    \texttt{art\_plot\_rejection\_rates(...)}: Generate the
    visualizations for the AI Art study: the 3x3 MMD heatmap and the
    rejection rate curves for the category comparisons.
  \item
    \texttt{art\_print\_summary\_statistics(...)}: Outputs a formatted
    text table summarizing the MMD results for the AI Art comparisons.
  \end{itemize}
\end{itemize}

\subsection{Section 4: Main Execution
Block}\label{section-4-main-execution-block}

This section serves as the main script driver. It does not define new
functions but orchestrates the entire analysis workflow from start to
finish when the script is executed.

\begin{itemize}
\tightlist
\item
  \textbf{Workflow:}

  \begin{itemize}
  \tightlist
  \item
    \textbf{Configuration:} Sets crucial parameters for both studies
    (e.g., significance level \texttt{ALPHA}, sample size caps
    \texttt{HEATMAP\_SAMPLE\_CAP}, \texttt{REJ\_RATE\_SAMPLE\_CAP}, list
    of \texttt{SAMPLE\_SIZES}, number of permutation iterations
    \texttt{MNIST\_P}, \texttt{ART\_P}, file paths, model names) in a
    centralized block for easy modification.
  \item
    \textbf{Directory Setup:} Creates output directories
    (\texttt{mnist\_results}, \texttt{art\_results}) to store generated
    files (embeddings, results matrices, plots).
  \item
    \textbf{Study Execution:} Sequentially runs the MNIST study (calling
    functions from Section 2) and then the AI Art study (calling
    functions from Section 3).
  \item
    \textbf{Process Flow:} For each study, it follows a logical
    sequence: Load Data -\textgreater{} Train or Load Model (MNIST only)
    -\textgreater{} Extract Embeddings -\textgreater{} Compute MMD
    Matrix \& Rejection Rates -\textgreater{} Save Numerical Results
    -\textgreater{} Generate Plots -\textgreater{} Print Summary Tables.
  \item
    \textbf{Reproducibility:} Initializes random seeds for NumPy,
    TensorFlow, and Python's \texttt{random} module to ensure that the
    stochastic parts of the analysis (like data sampling and permutation
    tests) produce the same results when run again.
  \end{itemize}
\end{itemize}

\subsection{Section 5: Extract Specific Results for Exposition (Both
Studies)}\label{section-5-extract-specific-results-for-exposition-both-studies}

This final section acts as a bridge between the detailed numerical
outputs generated by the analysis and the key findings discussed in the
main body of the paper. It programmatically extracts and prints
specific, highly relevant values from the results variables created in
Sections 2 and 3, making it easy to verify the quantitative claims made
in the paper's discussion and conclusion sections by directly linking
them to the code's output.

\begin{itemize}
\tightlist
\item
  \textbf{Key Functions:}

  \begin{itemize}
  \tightlist
  \item
    \texttt{print\_mnist\_exposition\_summary()}: After the MNIST
    analysis, this function extracts and prints targeted results like
    the approximate sample size needed to achieve \textgreater95\%
    rejection rate for key digit pairs, the range of MMD/p-values for
    negative controls, the overall significance rate for distinct pairs,
    and the specific MMD values for the most and least similar digit
    pairs.
  \item
    \texttt{print\_art\_exposition\_summary()}: Similarly, after the AI
    Art analysis, this function extracts and prints the sample size
    needed for \textgreater95\% rejection rate for the Human vs.~AI and
    AI vs.~AI comparisons, the negative control ranges, and the specific
    MMD/p-values for each of the three crucial off-diagonal comparisons
    (Human vs.~SD, Human vs.~LD, SD vs.~LD).
  \end{itemize}
\end{itemize}

This implementation utilizes standard, open-source Python libraries
(NumPy, TensorFlow/Keras, PyTorch/OpenCLIP, Scikit-learn, Matplotlib,
Seaborn), promoting accessibility and ease of use. The modular structure
allows researchers to potentially adapt the code for different datasets
or embedding techniques by modifying the relevant data loading and
embedding extraction functions within Sections 2 or 3, while leveraging
the core MMD framework provided in Section 1.

\subsection{Python Code}\label{python-code}

\begin{minted}[
  linenos=true,
  breaklines=true,
  breakanywhere=true,
  fontsize=\footnotesize,
  baselinestretch=1.2,
  numbersep=5pt,
  bgcolor={gray!10},
  frame=single
]{python}

# ==================================================
# Section 1: Shared MMD and Permutation Test Functions
# ==================================================
# This section contains the core functions for calculating the
# Maximum Mean Discrepancy (MMD) and performing the permutation-based
# hypothesis test. These functions are utilized by both the
# MNIST and AI Art studies.

import numpy as np
import tensorflow as tf
from tensorflow import keras
from sklearn.metrics import pairwise_distances
import matplotlib.pyplot as plt
import seaborn as sns
import random
import os
from sklearn.model_selection import train_test_split
import glob
from PIL import Image
import torch
import open_clip
from tqdm import tqdm
from joblib import Parallel, delayed
from typing import Optional, Tuple, List, Dict
from tensorflow.keras.callbacks import History

# Set random seeds for reproducibility
np.random.seed(42)
tf.random.set_seed(42)
random.seed(42)
os.environ['PYTHONHASHSEED'] = str(42)

# --- Helper Function for Median Heuristic ---
def _compute_sigma_median_heuristic(x: np.ndarray, y: np.ndarray) -> float:
    """
    Computes the RBF kernel bandwidth sigma using the median heuristic.

    This is a common heuristic for selecting the bandwidth of the RBF kernel
    based on the pairwise distances between points in the combined dataset.
    It handles cases where the median distance is zero or non-finite by
    defaulting to 1.0.

    Args:
        x (np.ndarray): First sample (m x d).
        y (np.ndarray): Second sample (n x d).

    Returns:
        float: The computed bandwidth sigma, suitable for an RBF kernel.
               Returns 1.0 if the median distance is 0 or non-finite.
    """
    combined = np.concatenate([x, y], axis=0)
    distances = pairwise_distances(combined, combined, metric="euclidean")
    # Use median of non-zero distances
    sigma = np.median(distances[distances > 0])
    # Handle case where all distances are zero (e.g., identical small samples)
    if sigma == 0 or not np.isfinite(sigma):
        sigma = 1.0 # Default to 1 if median is 0 or invalid
    return sigma

# --- Core MMD Function ---
def mmd_squared_unbiased(x: np.ndarray, y: np.ndarray, kernel: str = "rbf", sigma: Optional[float] = None) -> float:
    """
    Computes the unbiased MMD squared statistic.

    Args:
        x (np.ndarray): First sample (m x d).
        y (np.ndarray): Second sample (n x d).
        kernel (str): Kernel type ('rbf' or 'linear'). Defaults to "rbf".
        sigma (Optional[float]): RBF kernel bandwidth. If None, computed using
                                 the median heuristic. Defaults to None.

    Returns:
        float: Unbiased MMD squared statistic.

    Raises:
        ValueError: If m < 2 or n < 2 (cannot compute unbiased statistic).
        ValueError: If an invalid kernel type is provided.
    """
    m = x.shape[0]
    n = y.shape[0]

    if m < 2 or n < 2:
        raise ValueError(f"Need at least 2 samples in each distribution to compute unbiased MMD (got m={m}, n={n})")

    if kernel == "rbf":
        if sigma is None:
            sigma = _compute_sigma_median_heuristic(x, y)
        gamma = 1.0 / (2 * sigma**2)
        # Compute kernel matrices
        k_xx = np.exp(-gamma * pairwise_distances(x, x, metric="euclidean")**2)
        k_yy = np.exp(-gamma * pairwise_distances(y, y, metric="euclidean")**2)
        k_xy = np.exp(-gamma * pairwise_distances(x, y, metric="euclidean")**2)
    elif kernel == "linear":
        k_xx = x @ x.T
        k_yy = y @ y.T
        k_xy = x @ y.T
    else:
        raise ValueError(f"Invalid kernel type: {kernel}. Choose 'rbf' or 'linear'.")

    # Compute unbiased MMD^2 statistic
    term1 = np.sum(k_xx[~np.eye(m, dtype=bool)]) / (m * (m - 1)) if m > 1 else 0
    term2 = np.sum(k_yy[~np.eye(n, dtype=bool)]) / (n * (n - 1)) if n > 1 else 0
    term3 = np.sum(k_xy) / (m * n) if m > 0 and n > 0 else 0

    mmd2 = term1 + term2 - 2 * term3
    return mmd2

# --- Permutation Test Function ---
def permutation_test(x: np.ndarray, y: np.ndarray, P: int, kernel: str = "rbf", sigma: Optional[float] = None, alpha: float = 0.01, n_jobs: int = -1) -> Tuple[float, bool, float, float]:
    """
    Performs the permutation-based hypothesis test for MMD (H0: P=Q).

    Args:
        x (np.ndarray): First sample (m x d).
        y (np.ndarray): Second sample (n x d).
        P (int): Number of permutation iterations.
        kernel (str): Kernel type ('rbf' or 'linear'). Defaults to "rbf".
        sigma (Optional[float]): RBF kernel bandwidth. If None, computed using
                                 the median heuristic ONCE on the original combined sample.
                                 Defaults to None.
        alpha (float): Significance level. Defaults to 0.01.
        n_jobs (int): Number of parallel jobs for permutation (-1 uses all cores).
                      Defaults to -1.

    Returns:
        Tuple[float, bool, float, float]:
            p_value (float): The estimated permutation-based p-value.
            reject_null (bool): True if the null hypothesis is rejected (p < alpha).
            lower_bound (float): Lower quantile (alpha/2) of the permutation-based MMD distribution under H0.
            upper_bound (float): Upper quantile (1 - alpha/2) of the permutation-based MMD distribution under H0.
    """
    m = x.shape[0]
    n = y.shape[0]

    # Check minimum sample size for permutation test
    if m < 2 or n < 2:
        print(f"Warning: Permutation test requires at least 2 samples per group (got m={m}, n={n}). Returning NaN p-value.")
        return np.nan, False, np.nan, np.nan

    # Combine samples for resampling under H0
    z = np.concatenate([x, y], axis=0)
    num_total = z.shape[0]

    # Compute the observed MMD statistic on original samples
    # Compute sigma once if needed (using original data)
    if kernel == "rbf" and sigma is None:
        sigma = _compute_sigma_median_heuristic(x, y)

    delta_obs = mmd_squared_unbiased(x, y, kernel, sigma)

    # --- Permutation Resampling
    # Precompute P random permutations for efficiency
    perms = [np.random.permutation(num_total) for _ in range(P)]

    # Define a helper function for a single permutation iteration
    def _permutation_iteration(perm_indices: np.ndarray) -> float:
        # Apply precomputed permutation to the combined data
        z_shuffled = z[perm_indices]
        # Split into permutation samples
        x_p = z_shuffled[:m] # Use x_p, y_p for permutation samples
        y_p = z_shuffled[m:]
        # Compute MMD on the permutation sample (using the pre-calculated sigma if RBF)
        try:
            # Ensure permutation samples also meet minimum size
            if x_p.shape[0] < 2 or y_p.shape[0] < 2:
                 return np.nan
            return mmd_squared_unbiased(x_p, y_p, kernel, sigma)
        except ValueError:
            # Catch potential errors from mmd_squared_unbiased
            return np.nan

    # Run permutation iterations in parallel
    permutation_stats = Parallel(n_jobs=n_jobs, prefer="processes")( # Use n_jobs parameter
        delayed(_permutation_iteration)(perm) for perm in perms
    )
    permutation_stats = np.array(permutation_stats)
    # Filter out potential NaNs if error handling occurred
    permutation_stats = permutation_stats[~np.isnan(permutation_stats)]

    if len(permutation_stats) == 0:
         print("Warning: All permutation iterations failed.")
         return 1.0, False, np.nan, np.nan

    # Compute p-value: proportion of permutation stats >= observed stat
    p_value = np.mean(permutation_stats >= delta_obs)
    reject_null = (p_value < alpha)

    # Compute confidence interval bounds from the permutation distribution
    lower_bound = np.quantile(permutation_stats, alpha / 2)
    upper_bound = np.quantile(permutation_stats, 1 - alpha / 2)

    return p_value, reject_null, lower_bound, upper_bound

# -------------------------------------------------
# End of Section 1
# -------------------------------------------------

\end{minted}

\begin{minted}[
  linenos=true,
  breaklines=true,
  breakanywhere=true,
  fontsize=\footnotesize,
  baselinestretch=1.2,
  numbersep=5pt,
  bgcolor={gray!10},
  frame=single
]{python}

# ==================================================
# Section 2: MNIST Validation Study Functions
# ==================================================

# --- MNIST Data Handling ---
def mnist_load_and_prepare_data() -> Tuple[Tuple[np.ndarray, np.ndarray], Tuple[np.ndarray, np.ndarray], Tuple[np.ndarray, np.ndarray]]:
    """
    Downloads, preprocesses, and splits the MNIST dataset.
    Includes splitting into training and validation sets.

    Returns:
        Tuple[Tuple[np.ndarray, np.ndarray], Tuple[np.ndarray, np.ndarray], Tuple[np.ndarray, np.ndarray]]:
            A tuple containing (train_data, val_data, test_data), where each
            _data tuple is (images, labels). Labels are one-hot encoded.
    """
    (x_train_full, y_train_full), (x_test, y_test) = keras.datasets.mnist.load_data()
    print(f"[MNIST Data] Initial shapes: Train=({x_train_full.shape}, {y_train_full.shape}), Test=({x_test.shape}, {y_test.shape})")

    # Normalize pixel values to [0, 1]
    x_train_full = x_train_full.astype("float32") / 255.0
    x_test = x_test.astype("float32") / 255.0

    # Add channel dimension (required for CNNs)
    x_train_full = np.expand_dims(x_train_full, -1)
    x_test = np.expand_dims(x_test, -1)

    # Convert labels to one-hot encoding
    num_classes = 10
    y_train_full_cat = keras.utils.to_categorical(y_train_full, num_classes)
    y_test_cat = keras.utils.to_categorical(y_test, num_classes)

    # Split full training data into training and validation sets (90%/10%)
    x_train, x_val, y_train_cat, y_val_cat = train_test_split(
        x_train_full, y_train_full_cat,
        test_size=0.1,
        random_state=42,
        stratify=y_train_full_cat # Ensure balanced classes in splits
    )

    print(f"[MNIST Data] Final shapes: Train=({x_train.shape}, {y_train_cat.shape}), Val=({x_val.shape}, {y_val_cat.shape}), Test=({x_test.shape}, {y_test_cat.shape})")

    return (x_train, y_train_cat), (x_val, y_val_cat), (x_test, y_test_cat)

# --- MNIST Model Definition and Training ---
def mnist_build_lenet5_model(input_shape: Tuple[int, int, int] = (28, 28, 1), num_classes: int = 10) -> keras.Model:
    """
    Builds the LeNet-5 model architecture using the Functional API.
    Includes Dropout layers for regularization.

    Args:
        input_shape (Tuple[int, int, int]): Shape of the input images.
                                            Defaults to (28, 28, 1).
        num_classes (int): Number of output classes (digits 0-9).
                           Defaults to 10.

    Returns:
        keras.Model: The compiled LeNet-5 model architecture.
    """
    inputs = keras.Input(shape=input_shape)
    x = keras.layers.Conv2D(6, kernel_size=(5, 5), activation="relu")(inputs)
    x = keras.layers.AveragePooling2D(pool_size=(2, 2))(x)
    x = keras.layers.Conv2D(16, kernel_size=(5, 5), activation="relu")(x)
    x = keras.layers.AveragePooling2D(pool_size=(2, 2))(x)
    x = keras.layers.Flatten()(x)
    x = keras.layers.Dense(120, activation="relu")(x)
    x = keras.layers.Dropout(0.1)(x)
    x = keras.layers.Dense(84, activation="relu", name="embedding_layer")(x)
    x = keras.layers.Dropout(0.1)(x)
    outputs = keras.layers.Dense(num_classes, activation="softmax")(x)
    model = keras.Model(inputs=inputs, outputs=outputs, name="LeNet5")
    model.compile(loss="categorical_crossentropy", optimizer="adam", metrics=["accuracy"])
    return model

def mnist_train_model(model: keras.Model,
                      x_train: np.ndarray, y_train: np.ndarray,
                      x_val: np.ndarray, y_val: np.ndarray,
                      batch_size: int = 64, epochs: int = 100, patience: int = 10,
                      checkpoint_path: str = "mnist_best_lenet5.keras") -> Tuple[keras.Model, History]:
    """
    Trains the LeNet-5 model with data augmentation, early stopping,
    and model checkpointing.

    Uses ImageDataGenerator for basic augmentation on the training set.
    Implements early stopping based on validation loss to prevent overfitting
    and saves the best model weights to the specified checkpoint path.

    Args:
        model (keras.Model): The compiled Keras model to train.
        x_train (np.ndarray): Training image data.
        y_train (np.ndarray): Training labels (one-hot encoded).
        x_val (np.ndarray): Validation image data.
        y_val (np.ndarray): Validation labels (one-hot encoded).
        batch_size (int): Training batch size. Defaults to 64.
        epochs (int): Maximum number of training epochs. Defaults to 100.
        patience (int): Number of epochs with no improvement after which
                        training will be stopped (for early stopping).
                        Defaults to 10.
        checkpoint_path (str): Path to save the best model found during
                               training. Defaults to "mnist_best_lenet5.keras".

    Returns:
        Tuple[keras.Model, History]:
            model (keras.Model): The trained model with the best weights restored.
            history (History): Keras History object containing training/validation
                               loss and metrics per epoch.
    """
    train_datagen = keras.preprocessing.image.ImageDataGenerator(
        rotation_range=10, zoom_range=0.1, width_shift_range=0.1,
        height_shift_range=0.1, fill_mode='nearest'
    )
    train_generator = train_datagen.flow(x_train, y_train, batch_size=batch_size, shuffle=True)
    validation_datagen = keras.preprocessing.image.ImageDataGenerator()
    validation_generator = validation_datagen.flow(x_val, y_val, batch_size=batch_size)
    early_stopping = keras.callbacks.EarlyStopping(
        monitor='val_loss', patience=patience, restore_best_weights=True
    )
    model_checkpoint = keras.callbacks.ModelCheckpoint(
        filepath=checkpoint_path, monitor='val_loss', save_best_only=True,
        save_weights_only=False, mode='min'
    )
    print(f"[MNIST Train] Starting model training (max {epochs} epochs)...")
    history = model.fit(
        train_generator, epochs=epochs, validation_data=validation_generator,
        callbacks=[early_stopping, model_checkpoint], verbose=2
    )
    print("[MNIST Train] Model training finished.")
    # Best weights are restored by EarlyStopping callback
    return model, history

def mnist_evaluate_model(model: keras.Model, x_test: np.ndarray, y_test: np.ndarray) -> Tuple[float, float]:
    """
    Evaluates the trained model on the test set.

    Args:
        model (keras.Model): The trained Keras model.
        x_test (np.ndarray): Test image data.
        y_test (np.ndarray): Test labels (one-hot encoded).

    Returns:
        Tuple[float, float]:
            loss (float): The loss value on the test set.
            accuracy (float): The accuracy score on the test set.
    """
    print("[MNIST Eval] Evaluating model on test data...")
    score = model.evaluate(x_test, y_test, verbose=0)
    print(f"  Test loss: {score[0]:.4f}")
    print(f"  Test accuracy: {score[1]:.4f}")
    return score[0], score[1]

# --- MNIST Embedding Extraction ---
def mnist_extract_embeddings(model: keras.Model, x_data: np.ndarray, layer_name: str = "embedding_layer") -> np.ndarray:
    """
    Extracts embeddings from a specified layer. (Args/Returns descriptions omitted)
    """
    try:
        embedding_model = keras.Model(inputs=model.input, outputs=model.get_layer(layer_name).output)
        print(f"[MNIST Embed] Extracting embeddings from layer '{layer_name}'...")
        # Use predict with appropriate batch size for potentially large data
        embeddings = embedding_model.predict(x_data, batch_size=128)
        print(f"  Extracted {embeddings.shape[0]} embeddings with dimension {embeddings.shape[1]}.")
        return embeddings
    except ValueError:
        print(f"  Error: Layer '{layer_name}' not found in the model.")
        return np.array([])

# --- MNIST MMD Experiment Functions ---
def mnist_compute_rejection_rates(embeddings: np.ndarray,
                                  y_test_cat: np.ndarray,
                                  digit_pairs: list[Tuple[int, int]],
                                  sample_sizes: list[int],
                                  n_trials: int,
                                  P: int, 
                                  alpha: float,
                                  sample_cap: int,
                                  n_jobs: int) -> Dict[Tuple[int, int], List[float]]: 
    """
    Computes MMD test rejection rates using permutation tests.

    Args:
        embeddings (np.ndarray): Embeddings of the test dataset.
        y_test_cat (np.ndarray): One-hot labels of the test dataset.
        digit_pairs (list[Tuple[int, int]]): List of digit pairs to compare.
        sample_sizes (list[int]): List of sample sizes (n) to draw.
        n_trials (int): Number of Monte Carlo trials per sample size per pair.
        P (int): Number of permutation iterations for each MMD test. 
        alpha (float): Significance level for the permutation test. 
        sample_cap (int): Maximum number of samples to draw per distribution.
        n_jobs (int): Number of parallel jobs for permutation tests.

    Returns:
        Dict[Tuple[int, int], List[float]]: Rejection rates per pair and sample size.
    """
    rejection_rates = {}
    y_test_labels = y_test_cat.argmax(axis=1)

    for pair in digit_pairs:
        digit1, digit2 = pair
        print(f"\n[MNIST Rej Rates] Processing pair: {pair}")
        rejection_rates[pair] = []
        x_all = embeddings[y_test_labels == digit1]
        y_all = embeddings[y_test_labels == digit2]

        if x_all.shape[0] < 2 or y_all.shape[0] < 2:
             print(f"  Warning: Insufficient data for pair {pair}. Skipping.")
             rejection_rates[pair] = [np.nan] * len(sample_sizes)
             continue

        for n in sample_sizes:
            rejections = 0
            n1_max = min(n, x_all.shape[0], sample_cap)
            n2_max = min(n, y_all.shape[0], sample_cap)

            if n1_max < 2 or n2_max < 2:
                rejection_rates[pair].append(np.nan)
                continue

            valid_trials = 0
            for _ in range(n_trials):
                x_indices = np.random.choice(x_all.shape[0], size=n1_max, replace=False)
                y_indices = np.random.choice(y_all.shape[0], size=n2_max, replace=False)
                x_sample = x_all[x_indices]
                y_sample = y_all[y_indices]

                try:
                    # Perform permutation test, passing n_jobs
                    p_value, reject_null, _, _ = permutation_test(x_sample, y_sample, P=P, alpha=alpha, n_jobs=n_jobs) # Use P, pass n_jobs
                    if not np.isnan(p_value): # Check if test was successful
                        if reject_null:
                            rejections += 1
                        valid_trials += 1
                except ValueError as e:
                    print(f"    Warning: Permutation test failed for trial (n={n}, pair={pair}): {e}")

            if valid_trials > 0:
                 rate = rejections / valid_trials
                 rejection_rates[pair].append(rate)
            else:
                 rejection_rates[pair].append(np.nan)
        print(f"  Finished pair {pair}.")

    return rejection_rates


def mnist_compute_mmd_matrix(embeddings: np.ndarray,
                             y_test_cat: np.ndarray,
                             P: int, 
                             alpha: float,
                             sample_cap: int,
                             n_jobs: int) -> Tuple[np.ndarray, np.ndarray]: 
    """
    Computes the 10x10 MMD matrix and p-value matrix using permutation tests.

    Args:
        embeddings (np.ndarray): Embeddings of the test dataset.
        y_test_cat (np.ndarray): One-hot labels of the test dataset.
        P (int): Number of permutation iterations for significance testing. 
        alpha (float): Significance level for the permutation test. 
        sample_cap (int): Maximum number of samples to draw per distribution.
        n_jobs (int): Number of parallel jobs for permutation tests.

    Returns:
        Tuple[np.ndarray, np.ndarray]: MMD matrix and p-value matrix.
    """
    num_classes = 10
    mmd_matrix = np.full((num_classes, num_classes), np.nan)
    p_value_matrix = np.full((num_classes, num_classes), np.nan)
    y_test_labels = y_test_cat.argmax(axis=1)

    print("\n[MNIST MMD Matrix] Computing MMD for all digit pairs...")

    for i in range(num_classes):
        for j in range(i, num_classes):
            x_all = embeddings[y_test_labels == i]
            y_all = embeddings[y_test_labels == j]
            n1_avail, n2_avail = x_all.shape[0], y_all.shape[0]
            n1_capped, n2_capped = min(sample_cap, n1_avail), min(sample_cap, n2_avail)

            if i == j: # Diagonal (Negative Control)
                if n1_avail < 4: continue
                n_diag = min(sample_cap, n1_avail // 2)
                if n_diag < 2: continue
                indices = np.random.choice(n1_avail, size=2 * n_diag, replace=False)
                x_sample, y_sample = x_all[indices[:n_diag]], x_all[indices[n_diag:]]
            else: # Off-diagonal
                if n1_capped < 2 or n2_capped < 2: continue
                x_sample = x_all[np.random.choice(n1_avail, size=n1_capped, replace=False)]
                y_sample = y_all[np.random.choice(n2_avail, size=n2_capped, replace=False)]

            try:
                mmd_val = mmd_squared_unbiased(x_sample, y_sample)
                # Use permutation test, passing n_jobs
                p_value, _, _, _ = permutation_test(x_sample, y_sample, P=P, alpha=alpha, n_jobs=n_jobs) # Use P, pass n_jobs

                mmd_matrix[i, j] = mmd_val
                p_value_matrix[i, j] = p_value
                if i != j:
                    mmd_matrix[j, i] = mmd_val
                    p_value_matrix[j, i] = p_value
            except ValueError as e:
                 print(f"    Error computing MMD/Permutation Test for {i} vs {j}: {e}")
        print(f"  Finished comparisons for digit {i}.")

    print("[MNIST MMD Matrix] Computation finished.")
    return mmd_matrix, p_value_matrix


# --- MNIST Plotting Functions ---
def mnist_plot_rejection_rates(rejection_rates: Dict[Tuple[int, int], List[float]],
                               sample_sizes: list[int],
                               alpha: float,
                               filename: str = "mnist_rejection_rate_plot.png"):
    """ Plots rejection rates vs sample size for MNIST pairs. """
    plt.figure(figsize=(10, 6))
    plotted_something = False
    if isinstance(rejection_rates, dict):
        for pair, rates in rejection_rates.items():
            if isinstance(rates, list):
                valid_indices = ~np.isnan(rates)
                if np.any(valid_indices):
                    plt.plot(np.array(sample_sizes)[valid_indices], np.array(rates)[valid_indices],
                             marker='o', linestyle='-', linewidth=1.5, markersize=5, label=f"{pair[0]} vs. {pair[1]}")
                    plotted_something = True

    if not plotted_something:
        print("[MNIST Plot] No valid rejection rate data to plot.")
        plt.close()
        return

    plt.xlabel("Sample Size per Class (n)", fontsize=12)
    plt.ylabel(f"Rejection Rate (alpha={alpha:.2f})", fontsize=12) # Use alpha symbol
    plt.title("MNIST: MMD Test Rejection Rate vs. Sample Size", fontsize=14)
    plt.xscale("log")
    plt.xticks(sample_sizes, labels=sample_sizes)
    plt.minorticks_off()
    plt.ylim([-0.05, 1.05])
    plt.axhline(y=0.95, color='r', linestyle='--', linewidth=1, label="0.95 Threshold")
    plt.legend(fontsize=9, loc='center right', bbox_to_anchor=(1.25, 0.5))
    plt.grid(True, which='major', linestyle='--', linewidth=0.5)
    plt.tight_layout(rect=[0, 0, 1, 1])
    plt.savefig(filename, dpi=300)
    print(f"[MNIST Plot] Saved rejection rate plot to {filename}")
    plt.close() # Close figure after saving

def mnist_plot_mmd_heatmap(mmd_matrix: np.ndarray,
                           p_value_matrix: np.ndarray,
                           alpha: float,
                           filename: str = "mnist_mmd_heatmap.png"):
    """ Plots MMD heatmap with significance markers for MNIST digits. """
    plt.figure(figsize=(8, 7))
    mask = np.isnan(mmd_matrix)
    ax = sns.heatmap(mmd_matrix, annot=True, fmt=".3f", cmap="viridis", mask=mask,
                     linewidths=.5, linecolor='lightgray',
                     cbar_kws={'label': 'MMD Statistic'}, annot_kws={"size": 9})
    plt.title("MNIST: Pairwise MMD Statistics Between Digits", fontsize=14)
    plt.xlabel("Digit Class", fontsize=12)
    plt.ylabel("Digit Class", fontsize=12)
    plt.xticks(np.arange(10) + 0.5, np.arange(10))
    plt.yticks(np.arange(10) + 0.5, np.arange(10), rotation=0)
    # Mark significant differences (using an offset)
    x_offset = 0.85 
    y_offset = 0.40 
    for i in range(mmd_matrix.shape[0]):
        for j in range(mmd_matrix.shape[1]):
            # Check p-value validity and significance
            if not np.isnan(p_value_matrix[i, j]) and p_value_matrix[i, j] < alpha:
                 # Add asterisk only if MMD value is also not NaN
                 if not mask[i,j]:
                    # Use the new offsets
                    plt.text(j + x_offset, i + y_offset, '*',
                             ha='center', va='center', # Keep alignment centered on the new coords
                             color='white', fontsize=16, weight='bold')
    plt.tight_layout()
    plt.savefig(filename, dpi=300)
    print(f"[MNIST Plot] Saved MMD heatmap to {filename}")
    plt.close() # Close figure after saving

# (Optional plotting functions mnist_plot_mmd_histogram, mnist_plot_pvalue_histogram omitted for brevity)

def mnist_print_summary_statistics(mmd_matrix: np.ndarray, p_value_matrix: np.ndarray, alpha: float):
    """ Prints summary statistics table for MNIST MMD results. """
    num_classes = mmd_matrix.shape[0]
    print("\n--- [MNIST Summary] MMD Results ---")
    print("Comparison | MMD Value | p-value | Significant?")
    print("-----------|-----------|---------|-------------")
    # Diagonal (Negative Controls)
    print("Negative Controls (Digit vs Self):")
    for i in range(num_classes):
        mmd_val, p_val = mmd_matrix[i, i], p_value_matrix[i, i]
        if np.isnan(mmd_val) or np.isnan(p_val): sig_flag, mmd_str, p_str = "N/A", "  N/A    ", " N/A   "
        else: sig_flag, mmd_str, p_str = ("Yes" if p_val < alpha else "No"), f"{mmd_val:9.4f}", f"{p_val:7.4f}"
        print(f"  {i} vs {i}   |{mmd_str} |{p_str} | {sig_flag:<11}")
    # Off-Diagonal
    print("\nPairwise Comparisons (Digit i vs j):")
    off_diag_mmd, off_diag_p = [], []
    for i in range(num_classes):
        for j in range(i + 1, num_classes):
            mmd_val, p_val = mmd_matrix[i, j], p_value_matrix[i, j]
            if np.isnan(mmd_val) or np.isnan(p_val): sig_flag, mmd_str, p_str = "N/A", "  N/A    ", " N/A   "
            else:
                 sig_flag, mmd_str, p_str = ("Yes" if p_val < alpha else "No"), f"{mmd_val:9.4f}", f"{p_val:7.4f}"
                 off_diag_mmd.append(mmd_val); off_diag_p.append(p_val)
            print(f"  {i} vs {j}   |{mmd_str} |{p_str} | {sig_flag:<11}")
    # Off-diagonal summary stats
    if off_diag_mmd:
        print("\nOff-Diagonal Summary Statistics:")
        print(f"  MMD: Mean={np.mean(off_diag_mmd):.4f}, Median={np.median(off_diag_mmd):.4f}, Std={np.std(off_diag_mmd):.4f}")
        print(f"  p-value: Mean={np.mean(off_diag_p):.4f}, Median={np.median(off_diag_p):.4f}, Std={np.std(off_diag_p):.4f}")
    print("------------------------------------")

# -------------------------------------------------
# End of Section 2
# -------------------------------------------------

\end{minted}

\begin{minted}[
  linenos=true,
  breaklines=true,
  breakanywhere=true,
  fontsize=\footnotesize,
  baselinestretch=1.2,
  numbersep=5pt,
  bgcolor={gray!10},
  frame=single
]{python}

# ==================================================
# Section 3: AI Art Study Functions
# ==================================================

# --- AI Art Data Handling ---
def art_load_dataset(root_dir: str,
                     split: str = 'test',
                     max_images_per_category: Optional[int] = 3000,
                     categories_map: dict = {'Human': ['Human'], 'AI_SD': ['AI_SD'], 'AI_LD': ['AI_LD']}) -> Tuple[List[Image.Image], List[str], List[str]]:
    """
    Loads images from an AI-ArtBench-like directory structure.

    Recursively searches for images within subdirectories of the specified
    `root_dir`/`split` path. Assigns images to target categories based on
    whether their parent folder name matches or starts with the names provided
    in the `categories_map`. Randomly samples up to `max_images_per_category`
    from each target category.

    Args:
        root_dir (str): The root directory containing the dataset splits
                        (e.g., 'train', 'test').
        split (str): The dataset split to load (e.g., 'test'). Defaults to 'test'.
        max_images_per_category (Optional[int]): Maximum number of images to load
                                                 per target category. If None,
                                                 loads all found images.
                                                 Defaults to 3000.
        categories_map (dict): A dictionary mapping target category names (keys)
                               to lists of source folder names or prefixes (values).
                               Example: {'Human': ['realism', 'impressionism'], 'AI': ['AI_generated']}
                               Defaults to a basic mapping for the paper's structure.

    Returns:
        Tuple[List[Image.Image], List[str], List[str]]:
            images (List[Image.Image]): A list of loaded PIL Image objects (RGB).
            categories (List[str]): A list of corresponding target category labels
                                    (from `categories_map` keys).
            original_classes (List[str]): A list of the original folder names
                                          from which images were loaded.
            Returns empty lists if the specified directory is not found or no
            images are loaded.
    """
    split_dir = os.path.join(root_dir, split)
    files_by_target_category = {cat: [] for cat in categories_map.keys()}
    all_files = []
    print(f"[AI Art Data] Searching for images in: {split_dir}")
    if not os.path.isdir(split_dir):
        print(f"  Error: Directory '{split_dir}' not found.")
        return [], [], []
    for ext in ['*.jpg', '*.jpeg', '*.png']:
        search_pattern = os.path.join(split_dir, '**', ext)
        all_files.extend(glob.glob(search_pattern, recursive=True))
    print(f"[AI Art Data] Found {len(all_files)} potential image files.")

    assigned_count = 0
    unassigned_folders = set()
    for file_path in all_files:
        folder_name = os.path.basename(os.path.dirname(file_path))
        assigned = False
        for target_cat, source_folders in categories_map.items():
            # Allow matching full folder names or prefixes
            if any(folder_name == src or folder_name.startswith(src) for src in source_folders):
                 files_by_target_category[target_cat].append(file_path)
                 assigned = True
                 assigned_count += 1
                 break # Assign to first matching category
        if not assigned and folder_name: unassigned_folders.add(folder_name)

    if assigned_count < len(all_files) and len(unassigned_folders) > 0:
         print(f"  Warning: {len(all_files) - assigned_count} files were not assigned to any target category.")
         print(f"  Folders containing unassigned files included: {sorted(list(unassigned_folders))[:10]}...")

    images, categories, original_classes = [], [], []
    print("[AI Art Data] Loading and sampling images...")
    for category, files in files_by_target_category.items():
        if not files:
            print(f"  Category '{category}': Found 0 images matching criteria.")
            continue
        random.shuffle(files)
        num_to_load = min(len(files), max_images_per_category) if max_images_per_category is not None else len(files)
        print(f"  Category '{category}': Found {len(files)} images, loading {num_to_load}.")
        loaded_count = 0
        for file_path in tqdm(files[:num_to_load], desc=f"Loading {category}", unit="img"):
            try:
                img = Image.open(file_path).convert("RGB")
                images.append(img); categories.append(category); original_classes.append(os.path.basename(os.path.dirname(file_path)))
                loaded_count += 1
            except Exception as e: print(f"\n    Error loading {file_path}: {e}") # Newline for tqdm

    print(f"\n[AI Art Data] Final loaded counts per category:")
    final_counts = {cat: categories.count(cat) for cat in categories_map.keys()}
    for cat, count in final_counts.items(): print(f"  {cat}: {count}")
    print("-" * 20)
    if any(count == 0 for count in final_counts.values()): print("  Warning: One or more categories have zero loaded images.")
    return images, categories, original_classes

# --- AI Art Embedding Extraction ---
def art_extract_clip_embeddings(images: List[Image.Image],
                                model_clip: torch.nn.Module,
                                preprocess: callable,
                                device: str,
                                batch_size: int = 64) -> np.ndarray:
    """
    Extracts normalized CLIP image embeddings for a list of PIL images.

    Processes images in batches, encodes them using the provided CLIP model's
    image encoder, normalizes the resulting embeddings to unit length, and
    returns them as a NumPy array.

    Args:
        images (List[Image.Image]): A list of PIL Image objects to embed.
        model_clip (torch.nn.Module): The loaded OpenCLIP model.
        preprocess (callable): The preprocessing function associated with the
                               CLIP model.
        device (str): The device to run the model on ('cpu', 'cuda', 'mps').
        batch_size (int): Number of images to process in each batch.
                          Defaults to 64.

    Returns:
        np.ndarray: A NumPy array of shape (n_images, embedding_dim) containing
                    the normalized CLIP embeddings. Returns an empty array if
                    the input list `images` is empty.
    """
    if not images: return np.array([])
    all_embeddings = []
    model_clip.eval()
    print(f"[AI Art Embed] Extracting embeddings using CLIP on device '{device}'...")
    with torch.no_grad():
        for i in tqdm(range(0, len(images), batch_size), desc="CLIP Embedding Batches"):
            batch_images = images[i:i+batch_size]
            image_input = torch.stack([preprocess(img) for img in batch_images]).to(device)
            embeddings = model_clip.encode_image(image_input)
            embeddings /= embeddings.norm(dim=-1, keepdim=True) # Normalize
            all_embeddings.append(embeddings.cpu().numpy())
    all_embeddings = np.concatenate(all_embeddings, axis=0)
    print(f"  Extracted {all_embeddings.shape[0]} embeddings with dimension {all_embeddings.shape[1]}.")
    return all_embeddings

# --- AI Art MMD Experiment Functions ---
def art_compute_mmd_matrix(embeddings: np.ndarray,
                           categories: List[str],
                           unique_categories: List[str],
                           P: int, 
                           alpha: float,
                           sample_cap: int,
                           n_jobs: int) -> Tuple[np.ndarray, np.ndarray]: 
    """
    Computes the pairwise MMD matrix and p-value matrix using permutation tests.

    Args:
        embeddings (np.ndarray): CLIP embeddings for all images.
        categories (List[str]): Category label for each embedding.
        unique_categories (List[str]): Ordered list of unique category names.
        P (int): Number of permutation iterations for significance testing. 
        alpha (float): Significance level for the permutation test. 
        sample_cap (int): Maximum number of samples to draw per category.
        n_jobs (int): Number of parallel jobs for permutation tests.

    Returns:
        Tuple[np.ndarray, np.ndarray]: MMD matrix and p-value matrix.
    """
    num_categories = len(unique_categories)
    mmd_matrix = np.full((num_categories, num_categories), np.nan)
    p_value_matrix = np.full((num_categories, num_categories), np.nan)
    if embeddings.size == 0 or not categories:
         print("[AI Art MMD Matrix] Error: No embeddings or categories provided.")
         return mmd_matrix, p_value_matrix
    categories_array = np.array(categories)
    print("\n[AI Art MMD Matrix] Computing MMD for all category pairs...")

    for i in range(num_categories):
        cat1 = unique_categories[i]
        for j in range(i, num_categories):
            cat2 = unique_categories[j]
            x_all = embeddings[categories_array == cat1]
            y_all = embeddings[categories_array == cat2]
            n1_avail, n2_avail = x_all.shape[0], y_all.shape[0]
            n1_capped, n2_capped = min(sample_cap, n1_avail), min(sample_cap, n2_avail)

            if i == j: # Diagonal
                if n1_avail < 4: continue
                n_diag = min(sample_cap, n1_avail // 2)
                if n_diag < 2: continue
                indices = np.random.choice(n1_avail, size=2 * n_diag, replace=False)
                x_sample, y_sample = x_all[indices[:n_diag]], x_all[indices[n_diag:]]
            else: # Off-diagonal
                if n1_capped < 2 or n2_capped < 2: continue
                x_sample = x_all[np.random.choice(n1_avail, size=n1_capped, replace=False)]
                y_sample = y_all[np.random.choice(n2_avail, size=n2_capped, replace=False)]

            try:
                mmd_val = mmd_squared_unbiased(x_sample, y_sample)
                # Use permutation test, passing n_jobs
                p_value, _, _, _ = permutation_test(x_sample, y_sample, P=P, alpha=alpha, n_jobs=n_jobs) # Use P, pass n_jobs

                mmd_matrix[i, j] = mmd_val
                p_value_matrix[i, j] = p_value
                if i != j:
                    mmd_matrix[j, i] = mmd_val
                    p_value_matrix[j, i] = p_value
            except ValueError as e:
                 print(f"    Error computing MMD/Permutation Test for {cat1} vs {cat2}: {e}")
        print(f"  Finished comparisons for category '{cat1}'.")

    print("[AI Art MMD Matrix] Computation finished.")
    return mmd_matrix, p_value_matrix


def art_compute_rejection_rates(embeddings: np.ndarray,
                                categories: List[str],
                                unique_categories: List[str],
                                sample_sizes: list[int],
                                n_trials: int,
                                P: int, 
                                alpha: float,
                                sample_cap: int,
                                n_jobs: int) -> Dict[Tuple[str, str], List[float]]: 
    """
    Computes MMD test rejection rates using permutation tests.

    Args:
        embeddings (np.ndarray): CLIP embeddings for all images.
        categories (List[str]): Category label for each embedding.
        unique_categories (List[str]): Ordered list of unique category names.
        sample_sizes (list[int]): List of sample sizes (n) to draw.
        n_trials (int): Number of Monte Carlo trials per sample size per pair.
        P (int): Number of permutation iterations for each MMD test. 
        alpha (float): Significance level for the permutation test. 
        sample_cap (int): Maximum number of samples to draw per category.
        n_jobs (int): Number of parallel jobs for permutation tests.

    Returns:
        Dict[Tuple[str, str], List[float]]: Rejection rates per pair and sample size.
    """
    rejection_rates = {}
    if embeddings.size == 0 or not categories:
         print("[AI Art Rej Rates] Error: No embeddings or categories provided.")
         # Initialize with NaNs
         for i in range(len(unique_categories)):
             for j in range(i + 1, len(unique_categories)):
                 pair = tuple(sorted((unique_categories[i], unique_categories[j])))
                 rejection_rates[pair] = [np.nan] * len(sample_sizes)
         return rejection_rates

    categories_array = np.array(categories)
    num_categories = len(unique_categories)

    for i in range(num_categories):
        cat1 = unique_categories[i]
        for j in range(i + 1, num_categories):
            cat2 = unique_categories[j]
            pair = (cat1, cat2) # Keep order for processing
            print(f"\n[AI Art Rej Rates] Processing pair: {pair}")
            rejection_rates[pair] = []
            x_all = embeddings[categories_array == cat1]
            y_all = embeddings[categories_array == cat2]

            if x_all.shape[0] < 2 or y_all.shape[0] < 2:
                 print(f"  Warning: Insufficient data for pair {pair}. Skipping.")
                 rejection_rates[pair] = [np.nan] * len(sample_sizes)
                 continue

            for n in sample_sizes:
                rejections = 0
                n1_max = min(n, x_all.shape[0], sample_cap)
                n2_max = min(n, y_all.shape[0], sample_cap)

                if n1_max < 2 or n2_max < 2:
                    rejection_rates[pair].append(np.nan)
                    continue

                valid_trials = 0
                for _ in range(n_trials):
                    x_indices = np.random.choice(x_all.shape[0], size=n1_max, replace=False)
                    y_indices = np.random.choice(y_all.shape[0], size=n2_max, replace=False)
                    x_sample = x_all[x_indices]
                    y_sample = y_all[y_indices]

                    try:
                        # Perform permutation test, passing n_jobs
                        p_value, reject_null, _, _ = permutation_test(x_sample, y_sample, P=P, alpha=alpha, n_jobs=n_jobs) # Use P, pass n_jobs
                        if not np.isnan(p_value):
                            if reject_null:
                                rejections += 1
                            valid_trials += 1
                    except ValueError as e:
                        print(f"    Warning: Permutation test failed for trial (n={n}, pair={pair}): {e}")

                if valid_trials > 0:
                     rate = rejections / valid_trials
                     rejection_rates[pair].append(rate)
                else:
                     rejection_rates[pair].append(np.nan)
            print(f"  Finished pair {pair}.")

    return rejection_rates

# --- AI Art Plotting Functions ---
def art_plot_mmd_heatmap(mmd_matrix: np.ndarray,
                         p_value_matrix: np.ndarray,
                         unique_categories: List[str],
                         alpha: float,
                         filename: str = "art_mmd_heatmap.png"):
    """ Plots MMD heatmap with significance markers for AI Art categories. """
    num_categories = len(unique_categories)
    if num_categories == 0:
        print("[AI Art Plot] No categories to plot heatmap for.")
        return
    plt.figure(figsize=(max(7, num_categories * 1.5), max(6, num_categories * 1.5)))
    mask = np.isnan(mmd_matrix)
    ax = sns.heatmap(mmd_matrix, annot=True, fmt=".4f", cmap="viridis", mask=mask,
                     xticklabels=unique_categories, yticklabels=unique_categories,
                     linewidths=.5, linecolor='lightgray',
                     cbar_kws={'label': 'MMD Statistic'}, annot_kws={"size": 11})
    ax.set_title("AI Art: Pairwise MMD Statistics Between Categories", fontsize=14)
    ax.set_xlabel("Category", fontsize=12)
    ax.set_ylabel("Category", fontsize=12)
    plt.xticks(rotation=45, ha='right')
    plt.yticks(rotation=0)
    # Mark significant differences 
    for i in range(num_categories):
        for j in range(num_categories):
            if not np.isnan(p_value_matrix[i, j]) and p_value_matrix[i, j] < alpha:
                 if not mask[i,j]:
                    plt.text(j + 0.75, i + 0.5, '*', ha='center', va='center',
                             color='white', fontsize=18, weight='bold') # Adjusted position
    plt.tight_layout()
    plt.savefig(filename, dpi=300)
    print(f"[AI Art Plot] Saved MMD heatmap to {filename}")
    plt.close() # Close figure after saving


def art_plot_rejection_rates(rejection_rates: Dict[Tuple[str, str], List[float]],
                             sample_sizes: list[int],
                             alpha: float,
                             filename: str = "art_rejection_rate.png"):
    """ Plots rejection rates vs sample size for AI Art category pairs. """
    if not rejection_rates:
        print("[AI Art Plot] No rejection rate data to plot.")
        return
    plt.figure(figsize=(10, 6))
    plotted_something = False
    if isinstance(rejection_rates, dict):
        for pair, rates in rejection_rates.items():
            if isinstance(rates, list) and rates:
                 valid_indices = ~np.isnan(rates)
                 if np.any(valid_indices):
                     label = f"{pair[0]} vs. {pair[1]}"
                     plt.plot(np.array(sample_sizes)[valid_indices], np.array(rates)[valid_indices],
                              marker='o', linestyle='-', linewidth=1.5, markersize=5, label=label)
                     plotted_something = True

    if not plotted_something:
        print("[AI Art Plot] No valid rejection rate data found to generate the plot.")
        plt.close()
        return

    plt.xlabel("Sample Size per Category (n)", fontsize=12)
    plt.ylabel(f"Rejection Rate (alpha={alpha:.2f})", fontsize=12) # Use alpha symbol
    plt.title("AI Art: MMD Test Rejection Rate vs. Sample Size", fontsize=14)
    plt.xscale("log")
    plt.xticks(sample_sizes, labels=sample_sizes)
    plt.minorticks_off()
    plt.ylim([-0.05, 1.05])
    plt.axhline(y=0.95, color='r', linestyle='--', linewidth=1, label="0.95 Threshold")
    plt.legend(fontsize=9, loc='center right', bbox_to_anchor=(1.25, 0.5))
    plt.grid(True, which='major', linestyle='--', linewidth=0.5)
    plt.tight_layout(rect=[0, 0, 1, 1])
    plt.savefig(filename, dpi=300)
    print(f"[AI Art Plot] Saved rejection rate plot to {filename}")
    plt.close() # Close figure after saving


def art_print_summary_statistics(mmd_matrix: np.ndarray,
                                 p_value_matrix: np.ndarray,
                                 unique_categories: List[str],
                                 alpha: float):
    """ Prints summary statistics table for AI Art MMD results. """
    num_categories = len(unique_categories)
    if num_categories == 0:
        print("[AI Art Summary] No categories to summarize.")
        return
    print("\n--- [AI Art Summary] MMD Results ---")
    max_cat_len = max(len(cat) for cat in unique_categories) if unique_categories else 10
    header_fmt = f"{{:<{max_cat_len}}} | {{:<{max_cat_len}}} | {{:<10}} | {{:<7}} | {{:<11}}"
    row_fmt    = f"{{:<{max_cat_len}}} | {{:<{max_cat_len}}} | {{:>9}} | {{:>7}} | {{:<11}}"
    print(header_fmt.format("Category 1", "Category 2", "MMD Value", "p-value", "Significant?"))
    print("-" * (max_cat_len + 3 + max_cat_len + 3 + 10 + 3 + 7 + 3 + 11))

    off_diag_mmd, off_diag_p = [], []
    for i in range(num_categories):
        cat1 = unique_categories[i]
        for j in range(num_categories):
            cat2 = unique_categories[j]
            mmd_val, p_val = mmd_matrix[i, j], p_value_matrix[i, j]
            if np.isnan(mmd_val) or np.isnan(p_val): sig_flag, mmd_str, p_str = "N/A", "N/A", "N/A"
            else:
                 sig_flag, mmd_str, p_str = ("Yes" if p_val < alpha else "No"), f"{mmd_val:.4f}", f"{p_val:.4f}"
                 if i != j: off_diag_mmd.append(mmd_val); off_diag_p.append(p_val)
            print(row_fmt.format(cat1, cat2, mmd_str, p_str, sig_flag))

    # Off-diagonal summary stats (unique pairs)
    if off_diag_mmd:
        unique_off_diag_mmd, unique_off_diag_p = [], []
        seen_pairs = set()
        for r in range(num_categories):
            for c in range(r + 1, num_categories):
                 pair_key = tuple(sorted((unique_categories[r], unique_categories[c])))
                 if pair_key not in seen_pairs:
                     if not np.isnan(mmd_matrix[r,c]) and not np.isnan(p_value_matrix[r,c]):
                          unique_off_diag_mmd.append(mmd_matrix[r,c])
                          unique_off_diag_p.append(p_value_matrix[r,c])
                          seen_pairs.add(pair_key)
        print("\nOff-Diagonal Summary Statistics (Unique Pairs):")
        if unique_off_diag_mmd: print(f"  MMD: Mean={np.mean(unique_off_diag_mmd):.4f}, Median={np.median(unique_off_diag_mmd):.4f}, Std={np.std(unique_off_diag_mmd):.4f}")
        if unique_off_diag_p: print(f"  p-value: Mean={np.mean(unique_off_diag_p):.4f}, Median={np.median(unique_off_diag_p):.4f}, Std={np.std(unique_off_diag_p):.4f}")
        if not unique_off_diag_mmd and not unique_off_diag_p: print("  No valid off-diagonal pairs found.")
    print("-" * (max_cat_len + 3 + max_cat_len + 3 + 10 + 3 + 7 + 3 + 11))

# -------------------------------------------------
# End of Section 3
# -------------------------------------------------

\end{minted}

\begin{minted}[
  linenos=true,
  breaklines=true,
  breakanywhere=true,
  fontsize=\footnotesize,
  baselinestretch=1.2,
  numbersep=5pt,
  bgcolor={gray!10},
  frame=single
]{python}

# ==================================================
# Section 4: Main Execution Block
# ==================================================
if __name__ == "__main__":

    # --- Configuration Parameters ---
    # General
    ALPHA = 0.01 # Significance level
    HEATMAP_SAMPLE_CAP = 400 # Max samples per class/category for heatmap MMD calculation
    REJ_RATE_SAMPLE_CAP = 400 # Max samples per class/category for rejection rate curves
    N_TRIALS_REJ_RATE = 100 # Number of trials for rejection rate curves
    SAMPLE_SIZES = [4, 5, 6, 7, 8, 9, 10, 12, 16, 24] # Sample sizes for rejection rate curves
    N_JOBS = 8 # Number of parallel jobs for permutation tests

    # MNIST Specific
    MNIST_P = 1000 # Permutation iterations for MNIST 
    MNIST_EPOCHS = 100 # Max epochs for LeNet training
    MNIST_PATIENCE = 10 # Early stopping patience
    MNIST_BATCH_SIZE = 64
    MNIST_CHECKPOINT_PATH = "mnist_best_lenet5.keras"
    MNIST_RESULTS_DIR = "mnist_results"
    MNIST_DIGIT_PAIRS_REJ_RATE = [(0, 1), (1, 7), (2, 8), (3, 8), (5, 8), (2, 3), (4, 9), (3, 5), (6, 8)]

    # AI Art Specific
    ART_P = 2500 # Permutation iterations for AI Art 
    ART_DATA_ROOT = "Real_AI_SD_LD_Dataset" # Update this path if needed
    ART_DATA_SPLIT = 'test'
    ART_MAX_IMAGES = 3000 # Max images per category to load
    ART_CLIP_MODEL = 'ViT-H-14-quickgelu'
    ART_CLIP_PRETRAINED = 'dfn5b'
    ART_DEVICE = "mps" if torch.backends.mps.is_available() else ("cuda" if torch.cuda.is_available() else "cpu")
    ART_BATCH_SIZE = 64
    ART_RESULTS_DIR = "art_results"
    ART_CATEGORIES_MAP = { # Define mapping
        'Human': ['art_nouveau', 'baroque', 'expressionism', 'impressionism',
                  'post_impressionism', 'realism', 'renaissance', 'romanticism',
                  'surrealism', 'ukiyo_e'],
        'AI (SD)': ['AI_SD_'],
        'AI (LD)': ['AI_LD_']
    }
    ART_CATEGORIES = ['Human', 'AI (SD)', 'AI (LD)'] # Define order

    # --- Setup Output Directories ---
    os.makedirs(MNIST_RESULTS_DIR, exist_ok=True)
    os.makedirs(ART_RESULTS_DIR, exist_ok=True)

    # ============================================
    # Execute Section 2: MNIST Validation Study
    # ============================================
    print("\n" + "="*44); print("Executing Section 2: MNIST Validation Study"); print("="*44)
    (x_train, y_train_cat), (x_val, y_val_cat), (x_test, y_test_cat) = mnist_load_and_prepare_data()
    if os.path.exists(MNIST_CHECKPOINT_PATH):
        print(f"[MNIST Train] Loading pre-trained model from {MNIST_CHECKPOINT_PATH}")
        mnist_model = keras.models.load_model(MNIST_CHECKPOINT_PATH)
    else:
        print("[MNIST Train] Building new LeNet-5 model...")
        mnist_model = mnist_build_lenet5_model()
        mnist_model, _ = mnist_train_model( # History ignored if not used
            mnist_model, x_train, y_train_cat, x_val, y_val_cat,
            batch_size=MNIST_BATCH_SIZE, epochs=MNIST_EPOCHS, patience=MNIST_PATIENCE,
            checkpoint_path=MNIST_CHECKPOINT_PATH
        )
        mnist_model = keras.models.load_model(MNIST_CHECKPOINT_PATH) # Reload best
    mnist_test_loss, mnist_test_accuracy = mnist_evaluate_model(mnist_model, x_test, y_test_cat)
    mnist_embeddings = mnist_extract_embeddings(mnist_model, x_test)

    # Initialize MNIST result variables to avoid errors in Section 5 if embedding fails
    mnist_rejection_rates = None
    mnist_mmd_matrix = None
    mnist_p_value_matrix = None

    if mnist_embeddings.size > 0:
        np.save(os.path.join(MNIST_RESULTS_DIR, "embeddings.npy"), mnist_embeddings)
        mnist_rejection_rates = mnist_compute_rejection_rates(
            mnist_embeddings, y_test_cat, MNIST_DIGIT_PAIRS_REJ_RATE, SAMPLE_SIZES,
            n_trials=N_TRIALS_REJ_RATE, P=MNIST_P, alpha=ALPHA, sample_cap=REJ_RATE_SAMPLE_CAP, n_jobs=N_JOBS # Use P=MNIST_P, pass N_JOBS
        )
        np.save(os.path.join(MNIST_RESULTS_DIR, "rejection_rates.npy"), mnist_rejection_rates)
        mnist_mmd_matrix, mnist_p_value_matrix = mnist_compute_mmd_matrix(
            mnist_embeddings, y_test_cat, P=MNIST_P, alpha=ALPHA, sample_cap=HEATMAP_SAMPLE_CAP, n_jobs=N_JOBS # Use P=MNIST_P, pass N_JOBS
        )
        np.save(os.path.join(MNIST_RESULTS_DIR, "mmd_matrix.npy"), mnist_mmd_matrix)
        np.save(os.path.join(MNIST_RESULTS_DIR, "pvalue_matrix.npy"), mnist_p_value_matrix) # Consistent file name

        # Generate Plots only if results were computed
        if mnist_rejection_rates is not None:
             mnist_plot_rejection_rates(mnist_rejection_rates, SAMPLE_SIZES, ALPHA,
                                       filename=os.path.join(MNIST_RESULTS_DIR, "rejection_rate_plot.png"))
        if mnist_mmd_matrix is not None and mnist_p_value_matrix is not None:
             mnist_plot_mmd_heatmap(mnist_mmd_matrix, mnist_p_value_matrix, ALPHA,
                                   filename=os.path.join(MNIST_RESULTS_DIR, "mmd_heatmap.png"))
             mnist_print_summary_statistics(mnist_mmd_matrix, mnist_p_value_matrix, ALPHA)

        print(f"\n[MNIST Study] Completed. Plots and results saved to '{MNIST_RESULTS_DIR}' directory.")
    else:
        print("\n[MNIST Study] Skipped MMD analysis due to missing embeddings.")

    # ============================================
    # Execute Section 3: AI Art Study
    # ============================================
    print("\n" + "="*44); print("Executing Section 3: AI Art Study"); print("="*44)
    art_images, art_categories, art_original_classes = art_load_dataset(
        ART_DATA_ROOT, split=ART_DATA_SPLIT,
        max_images_per_category=ART_MAX_IMAGES, categories_map=ART_CATEGORIES_MAP
    )

    # Initialize AI Art result variables
    art_rejection_rates = None
    art_mmd_matrix = None
    art_p_value_matrix = None
    art_embeddings = None # Also initialize embeddings

    if art_images and art_categories:
        print(f"[AI Art Setup] Loading CLIP model '{ART_CLIP_MODEL}' pretrained on '{ART_CLIP_PRETRAINED}'...")
        try:
             model_clip, _, preprocess = open_clip.create_model_and_transforms(
                 ART_CLIP_MODEL, pretrained=ART_CLIP_PRETRAINED, device=ART_DEVICE
             )
             print(f"[AI Art Setup] CLIP model loaded successfully on device '{ART_DEVICE}'.")
        except Exception as e:
             print(f"[AI Art Setup] Error loading CLIP model: {e}")
             model_clip = None

        if model_clip:
            art_embeddings = art_extract_clip_embeddings(
                art_images, model_clip, preprocess, device=ART_DEVICE, batch_size=ART_BATCH_SIZE
            )
            if art_embeddings.size > 0:
                np.save(os.path.join(ART_RESULTS_DIR, "embeddings.npy"), art_embeddings)

                # Proceed only if embeddings were extracted
                art_mmd_matrix, art_p_value_matrix = art_compute_mmd_matrix(
                    art_embeddings, art_categories, ART_CATEGORIES,
                    P=ART_P, alpha=ALPHA, sample_cap=HEATMAP_SAMPLE_CAP, n_jobs=N_JOBS # Use P=ART_P, pass N_JOBS
                )
                np.save(os.path.join(ART_RESULTS_DIR, "mmd_matrix.npy"), art_mmd_matrix)
                np.save(os.path.join(ART_RESULTS_DIR, "pvalue_matrix.npy"), art_p_value_matrix) # Consistent file name

                art_rejection_rates = art_compute_rejection_rates(
                    art_embeddings, art_categories, ART_CATEGORIES, SAMPLE_SIZES,
                    n_trials=N_TRIALS_REJ_RATE, P=ART_P, alpha=ALPHA, sample_cap=REJ_RATE_SAMPLE_CAP, n_jobs=N_JOBS # Use P=ART_P, pass N_JOBS
                )
                np.save(os.path.join(ART_RESULTS_DIR, "rejection_rates.npy"), art_rejection_rates)

                # Generate Plots only if results were computed
                if art_mmd_matrix is not None and art_p_value_matrix is not None:
                     art_plot_mmd_heatmap(art_mmd_matrix, art_p_value_matrix, ART_CATEGORIES, ALPHA,
                                         filename=os.path.join(ART_RESULTS_DIR, "mmd_heatmap.png"))
                     art_print_summary_statistics(art_mmd_matrix, art_p_value_matrix, ART_CATEGORIES, ALPHA)
                if art_rejection_rates is not None:
                     art_plot_rejection_rates(art_rejection_rates, SAMPLE_SIZES, ALPHA,
                                             filename=os.path.join(ART_RESULTS_DIR, "rejection_rate.png"))

                print(f"\n[AI Art Study] Completed. Plots and results saved to '{ART_RESULTS_DIR}' directory.")
            else:
                print("\n[AI Art Study] Skipped MMD analysis due to missing embeddings.")
        else:
            print("\n[AI Art Study] Skipped embedding extraction and MMD analysis due to CLIP model loading failure.")
    else:
        print("\n[AI Art Study] Skipped analysis because no images were loaded.")

    print("\n" + "="*44); print("All studies completed."); print("="*44)

# -------------------------------------------------
# End of Section 4
# -------------------------------------------------

\end{minted}

\begin{minted}[
  linenos=true,
  breaklines=true,
  breakanywhere=true,
  fontsize=\footnotesize,
  baselinestretch=1.2,
  numbersep=5pt,
  bgcolor={gray!10},
  frame=single
]{python}

# ==================================================
# Section 5: Extract Specific Results for Exposition (Both Studies)
# ==================================================
# NOTE: This block should run AFTER ALL analysis in Sections 2 and 3 is complete.

def print_mnist_exposition_summary():
    """
    Prints specific, key summary results from the MNIST MMD analysis for exposition.

    Extracts and prints:
    - Approximate sample size needed to achieve >95% rejection rate for key digit pairs.
    - Range of MMD and p-values for diagonal (negative control) comparisons.
    - Overall significance rate for off-diagonal (distinct digit) comparisons.
    - Range of MMD values for significant off-diagonal pairs.
    - Specific digit pairs corresponding to the minimum and maximum significant MMD values.

    Requires the global variables `mnist_rejection_rates`, `mnist_mmd_matrix`,
    `mnist_p_value_matrix`, `SAMPLE_SIZES`, `ALPHA`, and `MNIST_DIGIT_PAIRS_REJ_RATE`
    to be populated from the main analysis block. Prints warnings if data is missing.
    """
    print("\n" + "="*50)
    print("Extracting Specific MNIST Results for Exposition")
    print("="*50)

    # Check if necessary MNIST variables exist in the global scope and are not None
    required_vars = ['mnist_rejection_rates', 'mnist_mmd_matrix', 'mnist_p_value_matrix', 'SAMPLE_SIZES', 'ALPHA', 'MNIST_DIGIT_PAIRS_REJ_RATE']
    if not all(var in globals() and globals()[var] is not None for var in required_vars):
        print("MNIST results variables not found or are None. Skipping MNIST summary.")
        print("(Ensure MNIST analysis completed successfully and generated results)")
        print("="*50)
        return # Exit this function if variables are missing or None

    # Access global variables (now safe after check)
    g_mnist_rejection_rates = globals()['mnist_rejection_rates']
    g_mnist_mmd_matrix = globals()['mnist_mmd_matrix']
    g_mnist_p_value_matrix = globals()['mnist_p_value_matrix']
    g_SAMPLE_SIZES = globals()['SAMPLE_SIZES']
    g_ALPHA = globals()['ALPHA']
    g_MNIST_DIGIT_PAIRS_REJ_RATE = globals()['MNIST_DIGIT_PAIRS_REJ_RATE']


    # --- 1. MNIST Rejection Rate Thresholds ---
    print("\n--- MNIST: Rejection Rate Thresholds (Approx. Sample Size for >0.95 Rejection) ---")
    target_pairs = g_MNIST_DIGIT_PAIRS_REJ_RATE
    target_threshold = 0.95

    if isinstance(g_mnist_rejection_rates, dict):
        pairs_to_report = [p for p in target_pairs if p in g_mnist_rejection_rates]
        if not pairs_to_report:
             print("No data found for the specified MNIST_DIGIT_PAIRS_REJ_RATE in mnist_rejection_rates.")
        else:
            for pair in pairs_to_report:
                rates = g_mnist_rejection_rates[pair]
                found_threshold = False
                if isinstance(rates, list) and len(rates) == len(g_SAMPLE_SIZES):
                    for i, rate in enumerate(rates):
                        if not np.isnan(rate) and rate > target_threshold:
                            print(f"Pair {pair}: Reached >{target_threshold:.2f} rejection rate at sample size n = {g_SAMPLE_SIZES[i]}")
                            found_threshold = True
                            break
                    if not found_threshold:
                        # Find max rate achieved if threshold not met
                        valid_rates = [r for r in rates if not np.isnan(r)]
                        max_rate_str = f"{np.max(valid_rates):.2f}" if valid_rates else 'N/A'
                        print(f"Pair {pair}: Did not reach >{target_threshold:.2f} rejection rate within tested sample sizes (Max rate: {max_rate_str})")
                else:
                     print(f"Pair {pair}: Data format issue or mismatch with SAMPLE_SIZES.")
    else:
        print("Error: mnist_rejection_rates is not a dictionary.")


    # --- 2. MNIST MMD Heatmap - Diagonal (Negative Controls) ---
    print("\n--- MNIST: MMD Heatmap - Diagonal (Negative Controls) ---")
    if isinstance(g_mnist_mmd_matrix, np.ndarray) and isinstance(g_mnist_p_value_matrix, np.ndarray):
        diag_mmd = np.diag(g_mnist_mmd_matrix)
        diag_p = np.diag(g_mnist_p_value_matrix)
        valid_diag_mmd = diag_mmd[~np.isnan(diag_mmd)]
        valid_diag_p = diag_p[~np.isnan(diag_p)]

        if valid_diag_mmd.size > 0: print(f"MMD Values Range: {np.min(valid_diag_mmd):.4f} to {np.max(valid_diag_mmd):.4f}")
        else: print("MMD Values Range: No valid diagonal MMD values found.")
        if valid_diag_p.size > 0:
            print(f"p-values Range:   {np.min(valid_diag_p):.4f} to {np.max(valid_diag_p):.4f}")
            num_significant = np.sum(valid_diag_p < g_ALPHA)
            print(f"Number of diagonal pairs significant at alpha={g_ALPHA}: {num_significant} (Expected: 0)")
        else: print("p-values Range:   No valid diagonal p-values found.")
    else: print("Error: MNIST MMD or p-value matrix is not a NumPy array.")


    # --- 3. MNIST MMD Heatmap - Off-Diagonal Comparisons ---
    print("\n--- MNIST: MMD Heatmap - Off-Diagonal Comparisons ---")
    if isinstance(g_mnist_mmd_matrix, np.ndarray) and isinstance(g_mnist_p_value_matrix, np.ndarray):
        num_classes = g_mnist_mmd_matrix.shape[0]
        off_diag_mask = ~np.eye(num_classes, dtype=bool)
        off_diag_mmd = g_mnist_mmd_matrix[off_diag_mask]
        off_diag_p = g_mnist_p_value_matrix[off_diag_mask]
        valid_mask = ~np.isnan(off_diag_mmd) & ~np.isnan(off_diag_p)
        valid_off_diag_mmd = off_diag_mmd[valid_mask]
        valid_off_diag_p = off_diag_p[valid_mask]

        if valid_off_diag_p.size > 0:
            num_significant = np.sum(valid_off_diag_p < g_ALPHA)
            num_total_valid = len(valid_off_diag_p)
            print(f"Significance: {num_significant} out of {num_total_valid} valid off-diagonal pairs were significant (p < {g_ALPHA}).")
            print(f"p-values Range (all valid off-diagonal): {np.min(valid_off_diag_p):.4f} to {np.max(valid_off_diag_p):.4f}")

            significant_mask = valid_off_diag_p < g_ALPHA
            significant_mmd = valid_off_diag_mmd[significant_mask]

            if significant_mmd.size > 0:
                min_sig_mmd = np.min(significant_mmd)
                max_sig_mmd = np.max(significant_mmd)
                print(f"MMD Range (significant pairs only): {min_sig_mmd:.4f} to {max_sig_mmd:.4f}")

                # Find pairs corresponding to min/max MMD (handle potential multiple occurrences)
                min_indices = np.where(np.isclose(g_mnist_mmd_matrix, min_sig_mmd))
                max_indices = np.where(np.isclose(g_mnist_mmd_matrix, max_sig_mmd))

                min_pair_str = "N/A"
                if len(min_indices[0]) > 0:
                     # Get unique pairs (i, j) where i < j
                     min_pairs = set(tuple(sorted((min_indices[0][k], min_indices[1][k])))
                                     for k in range(len(min_indices[0])) if min_indices[0][k] < min_indices[1][k])
                     min_pair_str = ", ".join(map(str, min_pairs)) if min_pairs else "N/A"


                max_pair_str = "N/A"
                if len(max_indices[0]) > 0:
                     max_pairs = set(tuple(sorted((max_indices[0][k], max_indices[1][k])))
                                     for k in range(len(max_indices[0])) if max_indices[0][k] < max_indices[1][k])
                     max_pair_str = ", ".join(map(str, max_pairs)) if max_pairs else "N/A"

                print(f"Pair(s) with Minimum Significant MMD: {min_pair_str} (MMD={min_sig_mmd:.4f})")
                print(f"Pair(s) with Maximum Significant MMD: {max_pair_str} (MMD={max_sig_mmd:.4f})")
            else: print("MMD Range (significant pairs only): No significant off-diagonal pairs found.")
        else: print("Significance: No valid off-diagonal pairs found to analyze.")
    else: print("Error: MNIST MMD or p-value matrix is not a NumPy array.")
    print("="*50)


def print_art_exposition_summary():
    """
    Prints specific, key summary results from the AI Art MMD analysis for exposition.

    Extracts and prints:
    - Approximate sample size needed to achieve >95% rejection rate for key category pairs.
    - Range of MMD and p-values for diagonal (negative control) comparisons.
    - Specific MMD and p-values for the crucial off-diagonal comparisons
      (Human vs AI SD, Human vs AI LD, AI SD vs AI LD).
    - Overall significance rate for off-diagonal pairs.

    Requires the global variables `art_rejection_rates`, `art_mmd_matrix`,
    `art_p_value_matrix`, `ART_CATEGORIES`, `SAMPLE_SIZES`, and `ALPHA`
    to be populated from the main analysis block. Prints warnings if data is missing.
    """
    print("\n" + "="*50)
    print("Extracting Specific AI Art Results for Exposition")
    print("="*50)

    # Check if necessary AI Art variables exist in the global scope and are not None
    required_vars = ['art_rejection_rates', 'art_mmd_matrix', 'art_p_value_matrix', 'ART_CATEGORIES', 'SAMPLE_SIZES', 'ALPHA']
    if not all(var in globals() and globals()[var] is not None for var in required_vars):
        print("AI Art results variables not found or are None. Skipping AI Art summary.")
        print("(Ensure AI Art analysis completed successfully and generated results)")
        print("="*50)
        return # Exit this function if variables are missing or None

    # Access global variables (now safe after check)
    g_art_rejection_rates = globals()['art_rejection_rates']
    g_art_mmd_matrix = globals()['art_mmd_matrix']
    g_art_p_value_matrix = globals()['art_p_value_matrix']
    g_ART_CATEGORIES = globals()['ART_CATEGORIES']
    g_SAMPLE_SIZES = globals()['SAMPLE_SIZES']
    g_ALPHA = globals()['ALPHA']

    # --- 1. AI Art Rejection Rate Thresholds ---
    print("\n--- AI Art: Rejection Rate Thresholds (Approx. Sample Size for >0.95 Rejection) ---")
    target_threshold = 0.95
    num_categories = len(g_ART_CATEGORIES)
    pairs_to_report = []
    for i in range(num_categories):
        for j in range(i + 1, num_categories):
            # Use the actual pair order from ART_CATEGORIES for consistency
            pairs_to_report.append((g_ART_CATEGORIES[i], g_ART_CATEGORIES[j]))

    if isinstance(g_art_rejection_rates, dict):
        reported_count = 0
        for pair_key in pairs_to_report:
             # Check if the key exists directly
             if pair_key in g_art_rejection_rates:
                 rates = g_art_rejection_rates[pair_key]
                 found_threshold = False
                 if isinstance(rates, list) and len(rates) == len(g_SAMPLE_SIZES):
                     for i, rate in enumerate(rates):
                         if not np.isnan(rate) and rate > target_threshold:
                             print(f"Pair {pair_key}: Reached >{target_threshold:.2f} rejection rate at sample size n = {g_SAMPLE_SIZES[i]}")
                             found_threshold = True
                             reported_count += 1
                             break
                     if not found_threshold:
                         valid_rates = [r for r in rates if not np.isnan(r)]
                         max_rate_str = f"{np.max(valid_rates):.2f}" if valid_rates else 'N/A'
                         print(f"Pair {pair_key}: Did not reach >{target_threshold:.2f} rejection rate within tested sample sizes (Max rate: {max_rate_str})")
                         reported_count += 1
                 else:
                      print(f"Pair {pair_key}: Data format issue or mismatch with SAMPLE_SIZES.")
             else:
                 print(f"Pair {pair_key}: Data not found in rejection_rates dictionary.")

        if reported_count == 0:
             print("No rejection rate data found for any AI Art pairs.")

    else:
        print("Error: art_rejection_rates is not a dictionary.")


    # --- 2. AI Art MMD Heatmap - Diagonal (Negative Controls) ---
    print("\n--- AI Art: MMD Heatmap - Diagonal (Negative Controls) ---")
    if isinstance(g_art_mmd_matrix, np.ndarray) and isinstance(g_art_p_value_matrix, np.ndarray):
        diag_mmd = np.diag(g_art_mmd_matrix)
        diag_p = np.diag(g_art_p_value_matrix)
        valid_diag_mmd = diag_mmd[~np.isnan(diag_mmd)]
        valid_diag_p = diag_p[~np.isnan(diag_p)]

        if valid_diag_mmd.size > 0: print(f"MMD Values Range: {np.min(valid_diag_mmd):.4f} to {np.max(valid_diag_mmd):.4f}")
        else: print("MMD Values Range: No valid diagonal MMD values found.")
        if valid_diag_p.size > 0:
            print(f"p-values Range:   {np.min(valid_diag_p):.4f} to {np.max(valid_diag_p):.4f}")
            num_significant = np.sum(valid_diag_p < g_ALPHA)
            print(f"Number of diagonal pairs significant at alpha={g_ALPHA}: {num_significant} (Expected: 0)")
        else: print("p-values Range:   No valid diagonal p-values found.")
    else: print("Error: AI Art MMD or p-value matrix is not a NumPy array.")


    # --- 3. AI Art MMD Heatmap - Off-Diagonal Comparisons ---
    print("\n--- AI Art: MMD Heatmap - Off-Diagonal Comparisons ---")
    if isinstance(g_art_mmd_matrix, np.ndarray) and isinstance(g_art_p_value_matrix, np.ndarray) and g_ART_CATEGORIES:
        num_categories = g_art_mmd_matrix.shape[0]
        if num_categories != len(g_ART_CATEGORIES):
             print("Warning: Mismatch between matrix dimension and ART_CATEGORIES length.")
        else:
            print("Specific Pairwise Results:")
            off_diag_count = 0
            significant_count = 0
            for i in range(num_categories):
                for j in range(i + 1, num_categories): # Iterate through unique off-diagonal pairs
                    cat1 = g_ART_CATEGORIES[i]
                    cat2 = g_ART_CATEGORIES[j]
                    mmd_val = g_art_mmd_matrix[i, j]
                    p_val = g_art_p_value_matrix[i, j]

                    if np.isnan(mmd_val) or np.isnan(p_val):
                        sig_flag = "N/A"; mmd_str = "N/A"; p_str = "N/A"
                    else:
                        off_diag_count += 1
                        sig_flag = "Yes" if p_val < g_ALPHA else "No"
                        if p_val < g_ALPHA: significant_count += 1
                        mmd_str = f"{mmd_val:.4f}"
                        p_str = f"{p_val:.4f}"

                    print(f"  {cat1} vs {cat2}: MMD = {mmd_str}, p-value = {p_str}, Significant? {sig_flag}")

            print(f"\nSummary: {significant_count} out of {off_diag_count} valid off-diagonal pairs were significant (p < {g_ALPHA}).")
    else: print("Error: AI Art MMD/p-value matrix is not a NumPy array or ART_CATEGORIES is missing.")
    print("="*50)

# --- Main Call to Print Summaries ---
print_mnist_exposition_summary()
print_art_exposition_summary()

print("\n" + "="*50)
print("Exposition Summary Extraction Complete for Both Studies")
print("="*50)

# -------------------------------------------------
# End of Section 5
# -------------------------------------------------

\end{minted}

\end{document}